\documentclass[journal,onecolumn,draftclsnofoot,12pt]{IEEEtran}%

\usepackage{setspace}
\usepackage{lettrine}

\usepackage{graphics,
           psfrag,
           epsfig,
           amsthm,
           cite,
           amssymb,
           url,
           dsfont,
           subfigure,
           balance,
           enumerate,
           color,
           setspace
}
\usepackage{amsmath}

\usepackage{epstopdf}

\newtheorem{definition}{Definition}

\newtheorem{theorem}{Theorem}

\newtheorem{property}{Property}
\usepackage{ragged2e}
\newcommand{\eref}[1]{(\ref{#1})}
\newcommand{\sref}[1]{Section~\ref{#1}}

\newcommand{\cref}[1]{Constraint~\ref{#1}}
\newcommand{\thref}[1]{Theorem~\ref{#1}}

\newcommand{\tref}[1]{Table~\ref{#1}}
\newcommand{\algref}[1]{Algorithm~\ref{#1}}

\newcommand{\ignore}[1]{}

\IEEEoverridecommandlockouts
\usepackage{mathtools}
\usepackage{color}
\usepackage[utf8]{inputenc}
\usepackage[ruled, lined, longend]{algorithm2e}
\SetEndCharOfAlgoLine{}
\usepackage{color}
\begin{document}

\title{\vspace{-.8cm} Coalition Formation Game for Delay Reduction in Instantly Decodable Network Coding Assisted D2D Communications}

\author{
 \IEEEauthorblockN{Mohammed S. Al-Abiad, \textit{Student Member, IEEE}, Ahmed Douik, \textit{Student Member, IEEE}, and Md. Jahangir Hossain, \textit{Senior Member, IEEE} \vspace{-1.9cm}}

\thanks {
Mohammed S. Al-Abiad and Md. Jahangir Hossain are with the School of Engineering, The University of British Columbia, Kelowna, BC V1V 1V7, Canada (e-mail: m.saif@alumni.ubc.ca, jahangir.hossain@ubc.ca).

Ahmed Douik is with the Department of Electrical Engineering, California Institute of Technology, Pasadena, CA 91125 USA (e-mail: ahmed.douik@caltech.edu).
}}

\maketitle

\begin{abstract}
Consider a wireless broadcast device-to-device (D2D) network wherein users' devices are interested in receiving some popular files. Each user's device possesses part of the content which is acquired in previous transmissions and cooperates with others to recover the missing packets by exchanging Instantly Decodable Network Coding (IDNC) packets.\ignore{Previous works suggest reducing the communication time by assuming the existence of a central controller.} Recently, a distributed solution, relying on a non-cooperative game-theoretic formulation, has been proposed to reduce the communication time for fully connected D2D networks, i.e., single-hop D2D networks. In this paper, we develop a distributed game-theoretical solution to reduce the communication time for a more realistic scenario of a \textit{decentralized} and \textit{partially} connected, i.e., multi-hop, IDNC-enabled D2D network. The problem is modeled as a coalition game with cooperative-players  wherein the payoff function is derived so that decreasing individual payoff results in the desired cooperative behavior. Given the intractability of the formulation, the coalition game is relaxed to a coalition formation game (CFG) involving the formation of disjoint coalitions. A distributed algorithm relying on merge-and-split rules is designed for solving the relaxed problem. The effectiveness of the proposed solution is validated through extensive numerical comparisons against existing methods in the literature. \vspace{-.6cm}
\end{abstract}

\begin{IEEEkeywords}
\vspace{-.5cm}Coalition game, device-to-device networks, instantly decodable network coding, multimedia streaming, real-time applications.\vspace{-.2cm}
\end{IEEEkeywords}

\section{Introduction} \label{sec:I}
\vspace{-.3cm}\IEEEPARstart{T}{he} use of smartphones and data-hungry applications in radio access networks are increasing dramatically worldwide\ignore{\cite{1}}. This growth impacts the ability of traditional wireless networks to meet the required Quality-of-Service (QoS) for its users. Device-to-device (D2D) communication has been proposed as a candidate technology \cite{3a,2} to support a massive number of connected devices and possibly improve the data rate of the next-generation mobile networks \cite{2a,27}. The decentralized nature of D2D networks allows devices to communicate with other nearby devices over short-range and possibly more reliable links which is suitable for numerous applications in mobile networks. For example, in wireless cellular networks, D2D system enables mobile traffic offloading by user cooperations for content downloading and sharing. Using conventional centralized Point-to-Multi-Point (PMP) networks, e.g., cellular, Wi-Fi, and fog/cloud radio access networks (FRAN/CRAN), for content delivery would be excessively complicated and expensive.\ignore{ Because systems comprising a massive number of devices require a large number of access points. These access points are connected to a backhaul network and require handover algorithms to support mobility.
In other practical applications, e.g., ad-hoc networks, such as wireless sensor networks (WSNs), D2D systems enable fast and reliable data communications for route and topology discovery as well as sending control/emergency packets.
\ignore{To overcome the backhaul and mobility challenges, D2D  systems have become a critical technology in the last few years. Indeed, D2D systems enable automated vehicles to communicate with nearby cars, pedestrians, and roadside devices without relying on a central and error-prone controller \cite{5a}.} }

Wireless channels are prone to interference and fading which result in packet/data loss at the application level. A widely used algorithm for packet recovery problem is the Automatic Repeat reQuest (ARQ)\ignore{\cite{6a}}.\ignore{This simple algorithm relies on negative acknowledgments (NACK) for the targeted devices to declare an erasure and reattempt the transmission.} However, this simple algorithm is highly inefficient for broadcast applications. For example, consider that a base-station (BS) is required to deliver the set of packets $\{p_1, p_2, p_3\}$ to users $\{u_1, u_2, u_3\}$. Assume that after sequentially transmitting $\{p_1, p_2, p_3\}$, user $u_i$ is still missing packet $p_i$ for $1 \leq i \leq 3$. To complete the reception of all packets for all users, the BS needs at least $3$ uncoded transmissions. However, by using an erasure code, the BS can broadcast the binary XOR combination $p_1\oplus p_2\oplus p_3$ that requires a single transmission.

Different erasure codes have been proposed for various applications and diverse network settings to solve the packet recovery problem. For the aforementioned PMP wireless broadcast networks, Raptor codes \cite{RC}, and Random Linear Network Codes (RLNC) \cite{RLNC} achieve maximum network throughput. Despite being efficient and offering a low-complexity solution, Raptor codes and RLNC are not attractive techniques for real-time applications, such as video streaming, online gaming, and teleconferencing. These codes accumulate a substantial decoding delay, meaning that these codes do not allow progressive decoding. In particular, coded packets cannot be decoded to retrieve the original data until a large number of independent transmissions are received.

Instantly Decodable Network Coding (IDNC) has been proposed as a low-complexity solution to improve throughput while allowing progressive decoding of the received packets \cite{IDNC}. By relying solely on binary XOR operations, IDNC ensures fast and instantaneous decodability of the transmitted packets for their intended users. Therefore, IDNC has been the topic of extensive research, e.g., \cite{12m,13m,14m, 15mm, 16m}. It has been applied in several real-time broadcast applications wherein received packets need to be used at the application layer immediately to maintain a high QoS, e.g., relay-aided networks \cite{17m,18m} video-on-demand and multimedia streaming \cite{20m,23m, 24m, 25m}, and D2D-enabled systems \cite{26m,28m,29m,30m}. The potential of IDNC technique is manifold \cite{31m}.

All the aforementioned IDNC works, for both PMP and D2D networks, are centralized in a sense that they require a global coordinator, i.e., a BS or a cloud, to plan packet combinations and coordinate transmissions. For example, the authors of \cite{30m} considered the completion time minimization problem in a partially connected D2D FRANs. The problem is solved under the assumption that the fog is within the transmission range of all devices and has perfect knowledge of the network topology. The authors suggested that the fog selects transmitting devices and their optimal packet combinations and conveys the information to the users for execution. 

While the aforementioned centralized approaches provide a good performance for the decentralized system, it comes at a high computation cost at the cloud/fog units and high power consumption at each user.  Indeed, users need to send the status of all D2D channels to the central controller at each time slot. In addition, the cloud controller requires to know the downloading history of users for content delivery.\ignore{ Furthermore, these  centralized approaches are not feasible in some network configurations, including the roadside-to-vehicle and vehicle-to-vehicle (V2V) applications discussed earlier.} Recently, the authors in \cite{33m}, \cite{33e}  proposed a distributed solution for D2D networks that rely on a non-cooperative game-theoretic formulation. However, in such game models, each player makes its decisions individually and selfishly. Furthermore, the system is assumed to be fully connected, i.e., single-hop, which only selects one player to transmit at any time instance. The fully connected model is not only an idealist in which all players are connected, it also causes severe latency (delay) in the network.  Our work proposes a fully distributed solution for completion time minimization in a partially connected D2D network using coalition games \cite{25}. Thus, multiple and altruistic players transmit IDNC packets simultaneously. 

Due to the cooperative and altruistic decisions among players, coalition games have been used in different network settings to optimize different parameters \cite{25aa, N1, 26, 42a, N2}. For example, the tutorial in \cite{25aa} classified the coalition games and demonstrated the applications of coalition games in communication networks. The authors of \cite{N1} proposed a distributed game theoretical scheme for users' cooperation in wireless networks to maximize users' rate  while accounting the cost of cooperation.  The authors of \cite{42a} proposed a Bayesian coalitional game for coalition-based cooperative packet delivery. Recently, the authors of \cite{N2} suggested a constrained
coalition formation game for minimizing users' content uploading in D2D multi-hop networks. For packet recovery purpose, we employ coalition game and IDNC optimization in D2D multi-hop networks.

Our work considers D2D multi-hop networks comprising several single-interface devices distributed in a geographical area, and each device is partially connected to other devices.\ignore{ The considered channel allocation follows \cite{26m, 29m, 30m, 33m}, where devices can transmit using the same radio resource block, e.g., frequency.}  The packet recovery problem is motivated by real-time applications that tolerate only
low delays, i.e., multimedia streaming. In such applications, users' devices need to immediately exchange a set of packets, represented by a frame, between them with the minimum communication time. Our proposed model appears in different applications. For example, in current LTE system, where users at the edge of the service area or in dense urban areas often experience high degradation in the quality of signal from data centers due to channel impairments. Our proposed D2D distributed scheme would improve the  total communication time of such users by implementing short and reliable D2D communications. Moreover, in cell centers with low erasures, our proposed scheme would offload the cloud's resources, e.g.,
time, bandwidth, and the ability to serve more users. 

Motivated by the aforementioned discussions, our work solves the completion time reduction problem in partially connected D2D networks. To this end, we introduce a novel coalition game framework capturing the complex
interplays of instantly decodable network coding, transmitting user-receiving user associations, and a limited coverage zone of each user. The main contributions of this work can be summarized as follows.
\begin{enumerate}
\item We formulate the \textit{completion time minimization problem} in partially connected D2D networks and model it as a coalition game. We further demonstrate the difficulty of expressing the problem as a coalition game with non-transfer function (NTU) which motivates its relaxation to a \textit{coalition formation game} (CFG). 
\item We derive the rules for assigning players\footnote{Player and device are used interchangeably throughout this paper.}, selecting transmitting player, and finding optimal encoded IDNC packets for each disjoint altruistic coalition.
\item We propose a distributed algorithm based on merge-and-split rules and study its convergence analysis, stability, \ignore{robustness,}complexity, and communication overhead. 
\item We validate our theoretical finding using numerical simulations. Our numerical results reveal that our distributed scheme can significantly outperform existing centralized PMP and fully distributed methods. Indeed, for presented network setups, our coalition formation game offers almost the same performance as the centralized FRAN scheme. 
\end{enumerate}

The rest of this paper is organized as follows. \sref{SMMM} introduces the system model and formulates the completion time minimization problem. Afterward, the problem is modeled as a coalition game and relaxed to a coalition formation game in \sref{COA}. The proposed distributed algorithm can be found in \sref{PS}, and its convergence analysis, stability, complexity, and communication overhead are provided in \sref{PA}.  \sref{SR} numerically tests the performance of the proposed method against existing schemes, and \sref{CC} concludes the paper.

\section{System Overview and Problem Formulation} \label{SMMM}

The considered network and IDNC models are introduced in \sref{sec:sub1} and \sref{sec:sub2}, respectively. The fully distributed completion time reduction problem in the considered network is formulated in \sref{PF}. \sref{sec:sub3} further shows through a simple example that the completion time problem is generally intractable, which motivates the coalition game formulation in \sref{COA}.

\subsection{Network Model and Parameters} \label{sec:sub1}

Consider a D2D-enabled wireless network consisting of $N$ users denoted by the set $\mathcal{U}=\{u_1, u_2, \ \cdots, u_N\}$. These users are interested in receiving a frame $\mathcal{P}=\{p_1, p_2, \ \cdots, p_M \}$ of $M$ packets. The size of the frame $\mathcal{P}$ depends on the size of the packet and size of content. Due to previous initial transmissions, from data centers or access points, each device holds a part of the frame $\mathcal{P}$. The side information of the $u$-th device is represented by the following sets.
\begin{itemize} 
\item The \textit{Has} set $\mathcal{H}_u$: Successfully received packets.
\item The \textit{Wants} set $\mathcal {W}_u=\mathcal{P} \setminus \mathcal{H}_u$: Erased/lost packets.
\end{itemize}

The side information of all players can be summarized in a binary $N \times M$ \textit{state matrix} $\mathbf{S}=[s_{up}]$ wherein the entry $s_{up}=0$ states that packet $p$ is successfully received by player $u$ and $1$ otherwise. In order for all users to obtain the whole frame $\mathcal{P}$ from D2D communications, we assume that each packet $p_i, 1 \leq i \leq M$ is received by at least one user. In other words, the sum of the rows $\sum_{u \in \mathcal{U}} s_{up} \geq 1$ for all packets $p \in \mathcal{P}$.

We consider a realistic multi-hop network topology. In such networks, battery-powered devices can only target the subset of devices in their coverage zone, denoted here by $\mathcal{C}_u$ of the $u$-th player.\ignore{ In network coding literature, multi-hop topology is referred to as partially connected D2D networks whereas single hop topology
is referred to as fully connected D2D networks.} The network topology can be captured by a unit diagonal symmetric $N \times N$ adjacency matrix $\mathbf{C}$ represents the connectivity of the players such that $\mathbf{C}_{u u^\prime} = 1$ if and only if $u^\prime \in \mathcal{C}_u$. We assume that no part of the network is disjoint, i.e., the matrix $\mathbf{C}$ is connected. Otherwise, the proposed algorithm is separately applied to each independent part of the network. Upon successful reception of a packet, each player send an error-free acknowledgment (ACK) to all players in its coverage zone to update their side information matrix. 

We focus only on upper layer view of the network, where network coding scheme is performed at the
network-layer and the physical-layer is abstracted by a memory-less erasure channel. This
abstraction is widely used in network coding literature, where a packet is either perfectly
received or completely lost with certain average probability  \cite{12m}, \cite{14m}, \cite{26m}, \cite{29m, 30m, 31m, 33m, 33e}, \cite{N3}.  Therefore, the physical channel between players $u$ and $u'$ is modeled by a Bernoulli random variable whose mean
$\sigma_{uu'}$ indicates the packet erasure probability from player $u$ to player $u'$. We assume that these probabilities remain constant during the transmission of a single packet $p_i \in \mathcal{P}$ and they are known to all devices. However, due to the channel's asymmetry and the difference in the transmit powers of both devices $u$ and $u^\prime$, the equality of $\sigma_{u u^\prime}$ and $\sigma_{u^\prime u}$ is not guaranteed.

We consider a slowly changing network topology, in which players have fixed locations during the IDNC packet transmission and change from one transmission to another transmission.
However, after one transmission, the devices can move and all the network variables will be updated, and our model, i.e., the coalition formation solution, can be used with updated network parameters. It is important to note that in single-hop networks, each player is
connected to all other players in the network, and hence, it precisely knows the side information of all other players. To avoid any collision in the network, only one player is allowed to transmit
an encoded packet in one hop at any time slot. Clearly, this causes severe latency, i.e., delay, in delivering packets to all players. In multi-hop networks, multiple players are allowed to transmit
encoded packets simultaneously. This
results in targeting many players, and thus makes the delivery of packets to the players faster.

\subsection{Instantly Decodable Network Coding Model} \label{sec:sub2}

\ignore{IDNC is a network coding scheme that suits most of the critical-time applications in which users need immediate decoding of the received packets. As mentioned earlier,}IDNC encodes packets through binary XOR operations. Let $\kappa \subset \mathcal{P}$ be an XOR combination of some packets in $\mathcal{P}$. The transmission of the combination $\kappa$ is beneficial to the $u$-th user, in a sense that it allows the $u$-th user to retrieve one of its missing packets, if and only if the combination contains a single packet from $\mathcal{W}_u$. In that case, the user $u$ can XOR the combination $\kappa$ with $\kappa \cap \mathcal{H}_u$ to obtain its missing packet. Hence, we say that the user $u$ is \emph{targeted} by the transmission $\kappa$.

Let $\mathcal{A}^{(t)} \subset \mathcal{U}$ denote the set of transmitting players at the $t$-th transmission and $\underline{\rm\kappa}^{(t)}(\mathcal{A})=(\kappa_1, \ \cdots, \ \kappa_{|\mathcal{A}^{(t)}|})$ denote the packet combinations to be sent by users in $\mathcal{A}^{(t)}$. For notation simplicity, the time index $t$ is often omitted when it is clear from the context. Similar to \cite{29m}, \cite{30m}, \cite{33m}, \cite{33e}, we  consider 
players use the same frequency band and transmit encoded packets simultaneously. Thus, players located in the intersection of the coverage zone of multiple transmitting players experience
collision at the network layer and no packets can be decoded. Considering the interference of transmissions caused by other players to the set of transmitting players in partially connected D2D networks can be pursued in a future work. Therefore, player $u^\prime$ is targeted by the transmission from the $u$-th player if and only if it can receive the transmission and the packet combination contains a single file from $\mathcal {W}_{u^\prime}$. Let $\underline{\rm\tau}({\underline{\rm\kappa}(\mathcal{A})})=(\tau_1, \ \cdots, \ \tau_{| \mathcal{A}|})$ denote the set of targeted players by the transmitting players wherein $u^\prime \in \tau_u({\underline{\rm\kappa}(\mathcal{A})})$ implies that $|\mathcal{W}_{u^\prime} \cap {{\rm\kappa}_u(\mathcal{A})}| = 1$ and $\{u^\prime\} \cap \mathcal{C}_u \cap \mathcal{C}_{u''} = \delta_{u{u''}} \{u^\prime\}$ for all transmitting players $u'' \in \mathcal{A}$ wherein $\delta_{u{u''}}$ is the Kronecker symbol.

\begin{definition} 
The individual completion time $\mathcal {T}_u$ of the $u$-th player is the number of transmissions required until it gets all packets in $\mathcal{P}$. The overall completion time $\mathcal {T} =\max _{u \in\mathcal{U}}\{\mathcal{T}_u\}$ represents the time required until all the players get all the packets.
\end{definition}

We use IDNC to minimize the completion time required to complete the reception of all packets for all users in the partially connected D2D network. Given that the direct minimization of the completion time is intractable \cite{31m}, we follow \cite{16m} in reducing the completion time by controlling the decoding delay.

\begin{definition} 
The decoding delay $\mathcal{D}_u$ of player $u$ increases by one unit if and only if the player still wants packets, i.e., $\mathcal{W}_i \neq \varnothing$, and receives a combination that does not allow it to reduce the size of its Wants set. The decoding delay $\mathcal{D}$ is the sum of all individual delays.
\end{definition}

\subsection{Completion Time Minimization Problem Formulation}\label{PF}

In this subsection, we formulate the distributed completion time reduction problem in IDNC-enabled D2D network. Let $\underline{\rm N}$ be a binary vector of size $N$ whose $u$-th index is $1$ if player $u$ has non-empty \textit{Wants} set, i.e., $\mathcal{W}_u\neq \varnothing$ and $0$ otherwise, and let $\underline{\rm{\overline{\rm{\tau}}}}(\underline{\kappa}(\mathcal{A}))=\underline{1}-\underline{\rm \tau}(\underline{\kappa}(\mathcal{A}))$ be the set of the non-targeted players by the encoded packets $\underline{\kappa}(\mathcal{A})$. The different erasure occurrences at the $t$-th time slot are denoted by $\boldsymbol{\omega}:\mathbb{Z}_+ \rightarrow \{0,1\}^{N\times N}$ with $\boldsymbol{\omega}(t)=[Y_{u u^\prime}]$, for all $(u,u^\prime)\in \mathcal{U}^2$, where $Y_{u u^\prime}$ is a Bernoulli random variable equal to $0$ with probability $\sigma_{u u^\prime}$. 

Let $\underline{\rm a}_t=(a_t^{[1]}, a_t^{[2]},\ \cdots, \ a_t^{[N]})$ be a binary vector of length $N$ whose $a_t^{[u]}$-th element is equal to $1$ if player $u$ is transmitting, i.e., $\|\underline{\rm a}\|_1=|\mathcal{A}|$. Likewise, let $\underline{\rm \mathcal{D}}(\underline{\rm a}_t)$ be the decoding delay
experienced by all players in the $t$-th recovery round. In particular, $\underline{\rm \mathcal{D}}(\underline{\rm a}_t)$ is a metric quantifies
the ability of the transmitting players to generate innovative packets for all the targeted players. This metric increases by one unit for each player that still wants packets and successfully receives a nonuseful
transmission from any transmitting player in $\mathcal{A}$ or for a transmitting player that still wants some packets.  Let $\underline{\rm \mathcal{I}}=(\mathcal{I}^{[1]}, \mathcal{I}^{[2]},\ \cdots, \ \mathcal{I}^{[N]})$ be a binary vector of size $N$ whose $\mathcal{I}^{[u]}$ entry is $1$ if player $u$ is hearing more than one transmission from the set $\mathcal{A}$, i.e., $u \in \mathcal{C}_{u^\prime }\cap \mathcal{C}_{u''}$ where $u^\prime \neq u'' \in \mathcal{A}$ and $0$ otherwise, and let $\underline{\rm \mathcal{O}}=(\mathcal{O}^{[1]}, \mathcal{O}^{[2]},\ \cdots, \ \mathcal{O}^{[N]})$ be a binary vector of size $N$ whose $\mathcal{O}^{[u]}$ element is $1$ if player $u$ is out of transmission range of any player in $\mathcal{A}$, i.e., $u \notin \mathcal{C}_{u^\prime}, \forall ~u^\prime \in \mathcal{A}$ and $0$ otherwise.  

Given the above configurations, the overall decoding delays $\underline{\rm \mathbb{D}}(\underline{\rm a}_t)$ experienced by all players, since the beginning of the recovery phase until the $t$-th transmission, can be expressed as follows.
\begin{align} \label{eq3}
\underline{\rm \mathbb{D}}(\underline{\rm a}_t) =\underline{\rm \mathbb{D}}(t-1) +
\begin{cases}
\underline N \hspace{2cm} & \mbox{if}~~\|\underline{\rm a}_t\|_1=0\\
\underline{ \mathcal{I}}+\underline{ \mathcal{O}}+ \underline{\rm a}_t+\underline{\rm \mathcal{D}}(\underline{\rm a}_t) & \mbox{otherwise}.
\end{cases}
\end{align}

As mentioned, the completion time is a difficult and intractable metric to optimize. However, in network coding literature, such metric is approximated by the \emph{anticipated} completion time which can be computed at each time instant using the decoding delay. Using the decoding delay in \eref{eq3}, the anticipated completion time is defined as follows.
\begin{definition}
The anticipated individual completion time of the $u$-th player is defined by the following expression 
\begin{align}
\mathcal{T}_u(\underline{\rm a}_t) = \frac{|\mathcal{W}^{(0)}_u|+ \mathbb{D}_u(\underline{\rm a}_t)-\mathbb{E}[\sigma_u]}{1-\mathbb{E}[\sigma_u]}, \label{eq85}
\end{align}
where $|\mathcal{W}^{(0)}_u|$ is the Wants set of player $u$ at the beginning of the recovery phase and $\mathbb{E}[\sigma_u]$ is the expected erasure probability linking player $u$ to the other players.\end{definition}
Clearly, \eref{eq85} represents the
number of transmissions that are required to complete the transmission of all requested packets
in $\mathcal{P}$. In this context, completion time is intimately related to the throughput of the system. Throughput is measured as the number of cooperative D2D transmission rounds required by the players to download all their requested
packets.

The overall anticipated completion time can be written as $\mathcal{T}(\underline{\rm a}_t) =\max\limits_u (\mathcal{T}_u(\underline{\rm a}_t))=\|\underline{\mathcal{T}}(\underline{\rm a}_t)\|_\infty$. Therefore, the anticipated completion time minimization problem at the $t$-th transmission in IDNC-enabled D2D multi-hop network can be written as follows.
\begin{align}
\min_ {\substack{\underline{\rm a}_t\in \{0,1\}^N\\ \underline{\kappa}(\mathcal{A}) \in \{0,1\}^{M}}} \|\underline{\mathcal{T}}(\underline{\rm a}_t)\|_\infty. \label{eqct}
\end{align}

Unlike single-hop model that requires only
an optimization over a single transmitting player and its corresponding packet combination, a multi-hop model needs to select the set of transmitting players $\mathcal{A}$ and the optimal
encoded packets $\underline{\kappa}(\mathcal{A})$. As such,  the probability of increasing the anticipated completion time is minimized.

\subsection{Example of IDNC Transmissions in a Partially Connected D2D-enabled Network}\label{sec:sub3}

This section illustrates the aforementioned definitions and concepts with a simple example. Consider a simple partially connected D2D network containing $6$ players and a frame $\mathcal{P}=\{p_1,p_2,p_3,p_4\}$ as illustrated in Fig. \ref{fig1}. The side information of all players is given on the left part of Fig. \ref{fig1}, and the coverage zone of each player is represented by edges. For ease of analysis, we assume error-free transmissions. 

Assume that $u_1$ transmits the encoded packet $\kappa_1=p_3\oplus p_4$ to players $u_2$, $u_3$, $u_5$, and let $u_6$ transmit $\kappa_6=p_1\oplus p_4$ to players $u_4, u_5$ in the first time slot. Then, in the second time slot, $u_4$ transmits $\kappa_4=p_2$ to $u_6$, and $u_1$ transmits $\kappa_1=p_2\oplus p_4$ to players $u_2, u_5$. The decoding delay experienced by the different players is given as follows.
\begin{itemize}
\item Player $u_5$ experiences one unit delay as it is in the intersection of the coverage zone of $u_1$ and $u_6$. In other words, $u_5$ is in collision, i.e., $u_5 \in \underline{\rm \mathcal{I}}$. Thus, player $u_5$ would not be able to decode packet $\kappa_6$ transmitted by player $u_6$.
\item Player $u_6$ experiences one unit of delay as it is transmitting in the first time slot. 
\end{itemize} 
Under this scenario, we have the following assumption.
\begin{itemize}
\item First time slot: $\underline{\rm N}=(0~1~1~1~1~1)$, the set of transmitting players $\mathcal{A}^{(1)}=\{u_1, u_6\}=\underline{\rm a}_1=(1~0~0~0~0~ 1)$, the corresponding encoded packets $\underline{\rm\kappa}(\mathcal{A}^{(1)})=(\kappa_1, \kappa_6)$, and the set of targeted players $\underline{\rm\tau}({\underline{\rm\kappa}(\mathcal{A}^{(1)})})=(\tau_1, \tau_6)=\{(u_2, u_3), (u_4)\}$. The set of players that hearing more than one transmission $\underline{\rm \mathcal{I}}=(0~0~0~0~1~0)$, and the set of players that out of transmission range of any player in $\mathcal{A}^{(1)}$ is $\underline{\rm \mathcal{O}}=\underline{\rm 0}$. The decoding delay experienced by all players is $\underline{\rm \mathcal{D}}(\underline{\rm a}_1)=(0~0~0~0~1~1)$. The accumulative decoding delay is $\underline{\rm \mathbb{D}}(\underline{\rm a}_1)=(0~0~0~0~1~1)$.
\item Second time slot: $\underline{\rm N}=(0~1~0~0~1~1)$, the set of transmitting players $\mathcal{A}^{(2)}=\{u_1,u_4\}=\underline{\rm a}_1=(1~0~0~1~0~ 0)$, the corresponding encoded packets $\underline{\rm\kappa}(\mathcal{A}^{(2)})=(\kappa_1, \kappa_4)$, and the set of targeted players $\underline{\rm\tau}({\underline{\rm\kappa}(\mathcal{A}^{(2)})})=(\tau_1,\tau_4)=\{(u_2, u_5), (u_6)\}$. The set of players hearing more than one transmission $\underline{\rm \mathcal{I}}=\underline{\rm 0}$, and the set of players that out of transmission range of any player in $\mathcal{A}^{(1)}$ is $\underline{\rm \mathcal{O}}=\underline{\rm 0}$. The decoding delay is $\underline{\rm \mathcal{D}}(\underline{\rm a}_2)=\underline{\rm 0}$ and the accumulative decoding delay $\underline{\rm \mathbb{D}}(\underline{\rm a}_2)=(0~0~0~0~1~1)$. 
\item The individual completion time of all players after the second transmission is \\$\mathcal{T}=(0~2~1~1~2~2)$. Thus, the maximum completion time is $2$ time slots which represents the overall completion time for all players to get their requested packets, i.e., $\underline{\rm N}=\underline{\rm 0}$. 
\end{itemize} 

\begin{figure}[t!]
\centering
\includegraphics[width=0.4\linewidth]{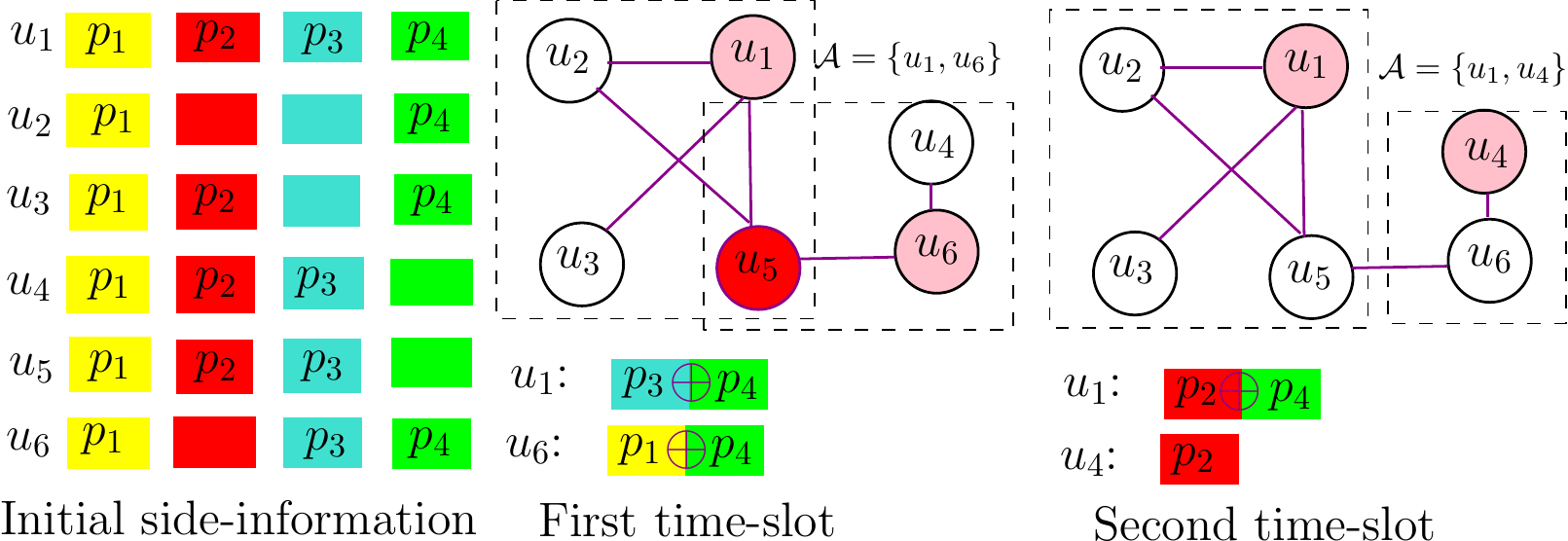}
\caption{A partially connected D2D network containing $6$ players and $4$ packets.}
\label{fig1}
\end{figure}

\ignore{The next section suggests formulating the above completion time optimization problem as a coalition formation game (CFG) to separate the following optimization parameters: the transmitting players and their optimal encoded packets. As such, a solution to \eref{eqct} can be obtained by finding a coalition structure using a fully distributed algorithm. }

\section{Distributed Completion Time Minimization as a Coalition Game}\label{COA}

This section models the completion time problem in IDNC-enabled D2D multi-hop networks using coalition games \cite{25}. Afterward, fundamental concepts in coalition games are defined and provided. These concepts are used in \sref{PS} to derive the distributed completion time reduction solution in a partially connected D2D network.

\subsection{Completion Time Minimization as a Coalition Game}

To mathematically model the aforementioned completion time problem, we use coalition game theory. In particular, the problem is modeled as a coalition game with a non-transferable utility (NTU)\cite{25}.
\begin{definition}
A coalition game with a \textit{non-transferable utility}
is defined as a pair ($\mathcal{U},\phi$), where $\mathcal{U}$ is the set of players consisting of $N$ devices and $\phi$ is a real function such that for every coalition $\mathcal{S}_s\subseteq \mathcal{U}$, $\phi(\mathcal{S}_s)$ is the payoff that coalition $\mathcal{S}_s$ receives which cannot be arbitrarily apportioned between its players.
\end{definition}

For the problem of cooperative D2D completion time among players, given any coalition $\mathcal{S}_s\subseteq \mathcal{U}$, we define $\phi(\mathcal{S}_s) = (\phi_1(\mathcal{S}_s), \ \cdots, \ \phi_{|\mathcal{S}|}(\mathcal{S}_s))$ as the tuple wherein element $\phi_u(\mathcal{S}_s)$ represents the payoff of player $u$ in coalition $\mathcal{S}_s$. Lets $|\mathcal{S}_s|$ represents the total number of players in $\mathcal{S}_s$. The $|\mathcal{S}|$-dimensional vector represents the family of real vector payoffs of coalition $\mathcal{S}_s$, which is denoted by $\underline{\phi}(\mathcal{S}_s)$.  As previously mentioned, for each coalition, we need to determine the transmitting player and its IDNC packet selection in order to minimize the increasing of the completion time. Consequently, by adopting the cooperative D2D completion model described in the previous section, the total payoff of any coalition $\mathcal{S}_s\subseteq \mathcal{U}$, $\forall s=\{1, \cdots,k\}$ is given by 
\begin{align}
\phi(\mathcal{S}_s) =\max\limits_u (\phi_u(\mathcal{S}_s))=\|\underline{\phi}(\mathcal{S}_s)\|_\infty,
\label{eq4}
\end{align}
where $\phi_u(\mathcal{S})$ is the payoff of player $u$ which is in our problem given by
\begin{align}
\phi_u(\mathcal{S}_s) = - \|\mathcal{T}_u(\underline{\rm a}_t)\|_\infty - \|{\mathbb{D}}_u(\underline a_t)-{\mathbb{D}}_u(\underline a_{t-1})\|_1.
\label{eq5}
\end{align}

The payoff function in \eref{eq4} represents the total payoff
that a coalition receives due to self-organize players. For
a player $u\in \mathcal{S}_s$, the first term in \eref{eq5} represents the maximum anticipated completion time among players in $\mathcal{S}_s$ that is defined in \eref{eq85}. Similarly, the second term in \eref{eq5} represents the augmentation of the sum decoding delay that is defined in \eref{eq3}. Therefore, players in coalitions prefer to increase the payoff in \eref{eq5} by minimizing the anticipated completion time through controlling the decoding delay.

\ignore{Since each player has its own unique payoff within a coalition $\mathcal{S}$, i.e., unique anticipated completion time and unique decoding delay, the payoff of coalition $\mathcal{S}$ cannot be arbitrarily apportioned between its players. Thus, the payoff function in \eref{eq5} is considered as a \textit{non transferable utility} (NTU) \cite{25}. Consequently, we immediately have the following
property.}

\begin{property} The proposed D2D completion time cooperative problem is modeled as a coalition game with NTU ($\mathcal{U},\phi$)
where $\mathcal{U}$ is the set of players and $\phi$ is the payoff function given by \eref{eq4}.
\end{property}

\begin{IEEEproof}
From the nature of definition 1 and definition 2, each player $u$ has its own unique anticipated completion time and decoding delay, and, thus, it has a unique payoff $\phi_u(\mathcal{S}_s)$ within a coalition $\mathcal{S}_s$. Therefore, the payoff function in \eref{eq4} cannot be arbitrarily apportioned between coalition's players. Thus \eref{eq4} is considered as an NTU. Further, the overall completion time is the maximum individual completion times of the players regardless of the coalition. In other words, the dependency of $\phi(\mathcal{S}_s)$ in any coalition structure is not only on packet recovery of players inside $\mathcal{S}_s$, but also on packet recovery outside $\mathcal{S}_s$, which concludes that the proposed game model is NTU game.
\end{IEEEproof}

Although cooperation generally reduces the payoffs of players \cite{25}, it is limited by inherent information exchange cost that needs to be paid by the players when acting cooperatively. Consequently, for any coalition $\mathcal{S}_s\subseteq\mathcal{U}$, players need to exchange information for cooperation, which is an increasing function of the coalition size. The problem becomes severe when all players are in the same coalition, i.e., grand coalition (GC). However, given the realistic scenario of a partially connected network where each device has limited coverage, it is highly likely that when attempting to form the GC, one of these scenarios would hold: 1) there exist a pair of players $u$, $u^\prime \in \mathcal{U}$ that are distant enough to receive packets from the set $\mathcal{A}$, thus they have no incentive to join the grand coalition, and 2) there exists a player $u \in \mathcal{U}$ with a payoff in GC $\phi_u(\mathcal{U}(t))$ that is greater than its payoff in any coalition $\phi_u(\mathcal{S}_s)$. Hence, this player has an incentive to deviate from the GC. 

Since we consider partially connected D2D networks, players would most likely form coalitions with their neighbors based on their preferences, which results in forming small coalitions' sizes, not large coalitions' sizes. In other words, the GC of all the players is \textit{seldom} formed. Therefore, the cost due to small coalition formations would not have a significant impact on the payoff functions. Subsequently, the proposed ($\mathcal{U},\phi$) game is classified as a coalition formation game  (CFG) \cite{25aa}, where players form several independent disjoint coalitions. Hence, classical solution concepts for coalition games, such as the core \cite{25}, may not be applicable for our problem. In brief, the proposed coalition game ($\mathcal{U},\phi$) is a CFG, where the objective is to offer an algorithm for forming coalitions.

\subsection{Coalition Formation Concepts}
This section recalls the fundamental concepts of coalition formation games that are used in the next section.
CFG, a subclass of coalition games, has been a topic of high interest in game theory research \cite{25aa, N1, N2}. The fundamental approach in coalition formation games is to allow players in the formation set to join or leave a coalition based on a well-defined and most suitable \textit{preference} for NTU games, i.e., \textit{Pareto Order}. Pareto Order is the basis of many existing
coalition formation concepts, e.g., the merge-and-split algorithm \cite{26}.\ignore{ In what follows, key definitions about coalition formations are introduced \cite{44aaaa}.}

\begin{definition} A coalition structure, denoted as $\Psi$, is defined as $\Psi=\{\mathcal{S}_1, \ \cdots, \mathcal{S}_k\}$ for $1 <|\mathcal{S}_k|< |\mathcal{U}|$ independent disjoint coalitions $\mathcal{S}_k$ of $\Psi$.
\end{definition}

One can see from definition 5 that different coalition structures may lead to different system payoffs as each coalition structure $\Psi$ has its unique payoff $\phi(\Psi)$. These differences in $\Psi$ and their corresponding payoffs $\phi(\Psi)$ are usually ordered through a comparison relationship. In the coalition game literature, e.g., \cite{26}, comparison relationships based on orders are divided into individual value orders and coalition value orders. Individual order implies that comparison is performed based on the players' payoffs. This is referred to as the Pareto Order. In particular, in such order, no player is willing to move to another coalition when at least one of the players in that coalition is worse off. In other words, the payoff of players would be worse off after the new player joins. This is known as selfish behavior. Coalition order implies that two coalition structures are compared based on the payoff of the coalitions in these coalition structures. This is known as a utilitarian order and is denoted by $\triangleright$. In other words, the notation $\Psi_2\triangleright \Psi_1$ means that $\phi(\Psi_1)>\phi(\Psi_2)$. Subsequently, the definition of the preference operator that considered in this paper is given as follows.
\begin{definition} A preference operator $\triangleright$ is defined for comparing two coalition structures $\Psi_1=\{\mathcal{S}_1, \ \cdots, \mathcal{S}_k\}$ and $\Psi_2=\{\mathcal{R}_1, \ \cdots, \mathcal{R}_m\}$ that are partitions of the same set of players $\mathcal{U}$. The notation $\Psi_2 \triangleright \Psi_1$ denotes that players in $\mathcal{U}$ are preferred to be in $\Psi_2$ than $\Psi_1$.
\end{definition}

\ignore{Given the above coalition formation concepts, next section gives the rules of forming a coalition in the above partially D2D connected network to provide a fully distributed algorithm.}


\section{Proposed Fully Distributed Solution} \label{PS}
This section derives the constraints of forming a coalition. These constraints represent the optimal players' associations, the transmitting player, and its optimal IDNC packet in a coalition. By the given constraints, our aim is to propose a distributed coalition formation algorithm relying on merge-and-split rules \cite{26}. 

\subsection{Coalition Formation Constraints}

Let $\mathcal{U}_s$ be the set of all associated players in coalition $\mathcal{S}_s$ and $\mathcal{N}_s$ the subset of $\mathcal{U}_s$ that have non-empty \textit{Wants} set. Let $\mathcal{M}_s$ be the subset of packets that in the \textit{Has} set of each player in $\mathcal{U}_s$, which defined as $\mathcal{M}_s=\bigcup_{u \in \mathcal{U}_s}\mathcal{H}_{u}$. Let $\mathtt{S}_s$ denote the set of all neighbor coalitions to coalition $\mathcal{S}_s$. For a coalition $\mathcal{S}_s$, the transmitting device $a^*_{s}$ is the one that can achieve the least expected increase in the completion time. 
According to the analysis available in \cite{33m,33e}, a transmitting device $a^{*}_s$ and its packet combination  $\kappa_{a_s^*}$ can be obtained by solving the following problem
\begin{align}\label{op7}
a_s^{*}=\operatorname*{arg\,max}\limits_{{\substack{\ a\in \mathcal{A}_s\setminus \mathcal{L}_s }}}|\mathcal{C}_{a}\cap\mathcal{N}_s|+\operatorname*{max}\limits_{\substack{\ \kappa_{a}\in \rm{\underline{\rm\kappa}(\mathcal{A}_s)}}} \sum\limits _{u\in\mathcal{L}_s\cap\tau(\kappa_{a})} \log\frac{1}{\sigma_{au}},
\end{align}
where $\mathcal{A}_s$ is the set of players in coalition $\mathcal{S}_s$ that are not in any coverage zone of all other players in $\mathtt{S}_s$ and $\mathcal{L}_s(t)$ is the set of critical players that can potentially increase the overall payoff of the coalition $\mathcal{S}_s$ before the $t$-th transmission. This set characterizes the players based on their anticiapted completion times to give them priority to be targeted in the next transmission. In other words, $\mathcal{L}_s(t)$ contains players that would potentially increase the maximum anticipated completion time if they are not targeted in
the next transmission. It can be define  mathematically as 
\begin{equation} 
\begin{split}
\label{CSG}
\mathcal {L}_s(t)= \Bigl\{u\in\ \mathcal{U}\cap \mathcal{N}_s \big| \mathcal{T}_u(\underline{a}_t-1)+\frac{1}{1-\mathbb{E}[\sigma_u]}
\geqslant \|\underline{\mathcal{T}}(\underline{a}_t-1)\|_\infty\Bigr\}.
\end{split}
\end{equation} The set of targeted players in coalition $\mathcal{S}_s$ when device $a_s^{*}$ transmits the combination $\kappa_{a_s^*}$ is
\begin{align}
\tau(\kappa_{a_s^{*}})=\left\{u \in \mathcal{S}_s \ \big||\kappa_{a_s^{*}} \cap \mathcal{W}_u| = 1~~\text{and}~~ \mathcal{C}_{a_s^* u} =1 \right\}.
\end{align}
With the aforementioned variable definitions, we can reformulate the completion time minimization problem in IDNC-based partially connected D2D network per coalition at each time instance as follows
\begin{subequations}\label{CF} 
\begin{align} 
&\min_ {{\substack{\ \underline{a}_t\in\{0,1\}^{|\mathcal{U}_s|} \\\underline{\kappa}\in\{0,1\}^{|\mathcal{M}_s|} }}}\phi(\mathcal{S}_s)\label{eq7aa} \\
& {\rm s. ~t.\ } |\tau(\kappa_{a_s^{*}})| \geqslant 1,\label{eq7bb} \\
&\tau(\kappa_{a_s^{*}}) \cap \tau(\kappa_{a_{s^\prime}^{*}}) =\varnothing, \forall~ a^*_s \neq a^*_{s^\prime} \in \mathtt{S}_s.\label{eq7dd}
 \end{align}
\end{subequations}
Constraint \eref{eq7bb} says that the number of targeted players in each coalition must be more than one to ensure that at each transmission at least a player is benefiting. Constraint \eref{eq7dd} states that all targeted players should not experience any collision. 

To find the optimal solution to the problem in \eref{CF}, we need to search over all the sets of optimal player-coalition associations, their different erasure patterns, players' actions and their optimal IDNC packets in one coalition. As pointed out in \cite{30m} for centralized fog system, this is a challenging problem. Further, the solution to \eref{CF} must go through the players' decisions to join/leave a coalition at each stage of the game. To seek a desirable solution to \eref{CF} that is capable of achieving significant completion time reduction, we propose to use a distributed algorithm relying on merge-and-split rules.

\subsection{A Distributed Coalition Formation Algorithm} \label{DC}
This section presents
a distributed coalition forming algorithm to obtain the minimum completion time of players. The key mechanism is to allow players in coalition formation process to make individual decisions for selecting potential neighbor coalitions at any game stage. We first define two rules of merge-and-split that allow the modification of $\Psi$ of the set $\mathcal{U}$ players as follows.

\begin{definition} (\textbf{Merge Operation}). Any set of coalitions $\{\mathcal{S}_1, \ \cdots, \mathcal{S}_k\}$ in $\Psi_1$ can be merged if and only if $(\bigcup^{k}\limits_{i = 1}\mathcal{S}_{i},\Psi_2) \triangleright (\{\mathcal{S}_1, \ \cdots, \mathcal{S}_k\},\Psi_1)$, where $\bigcup^{k}\limits_{i = 1}\mathcal{S}_{i}$ and $\Psi_2$ are the new set of coalitions and the new coalition structure after the merge operation, respectively.
\end{definition}

\begin{definition} (\textbf{Split Operation}). Any set of coalitions $\bigcup^{k}\limits_{i = 1}\mathcal{S}_{i}$ in $\Psi_1$ can be split if and only if $(\{\mathcal{S}_1, \ \cdots, \mathcal{S}_k\},\Psi_2) \triangleright (\bigcup^{k}\limits_{i = 1}\mathcal{S}_{i},\Psi_1)$, where $\{\mathcal{S}_1, \ \cdots, \mathcal{S}_k\}$ and $\Psi_2$ are the new set of coalitions and the new coalition structure after the split operation, respectively.
\end{definition}

The merge rule means that two coalitions merge if their merger would benefit not only the players in the merged coalition but also benefit the overall coalition structure value, i.e., the overall completion time. On the other hand, a coalition split into smaller ones if its splitter coalitions enhance at least the payoff of one player in that coalition. Therefore, using these two known rules, we present a distributed algorithm to solve the completion time minimization problem in \eref{eqct}. The proposed algorithm is broken into three steps as follows.

First, in $\Psi_\text{ini}$, players need to discover their neighbors by utilizing one of different known neighbor
discovery schemes, e.g., those used in wireless networks \cite{44aa}. For example, each player broadcasts a message consisting of two segments; each segment consists of one byte. While the first byte indicates the number of players in each player's coverage zone, the second byte indicates the completion time of that player. Further, players collect all aforementioned information, and the one who is connected to a large number of players, has a large \textit{Has} set, and not in the coverage zone of any player in any other coalitions. However, if such player does not exist, the size of the coalition is increased until that player exists.  To summarize, a transmitting player $a^*_{s}$ in coalition $s$ should satisfy \eref{eq7bb} and \eref{eq7dd} and can be obtained by solving problem \eref{op7}. \ignore{should satisfy these following conditions (C).
\begin{itemize} 
\item C1: Not in the critical set of players and connected to the largest possible players in the coalition, i.e., $a^{*}_{s}=\operatorname*{arg\,max}\limits_{{\substack{\ a\in \mathcal{A}_s\setminus \mathcal{L}_s }}}|\mathcal{C}_{a}\cap\mathcal{N}_s|$.
\item C2: Not in the coverage zone of any players in any other coalitions. In other words, $a^*_{s}$ can only target players in coalition $s$, i.e., $\tau(\kappa_{a_s^{*}}) \cap \tau(\kappa_{a_{s^\prime}^{*}}) =\varnothing, \forall~ a^*_s \neq a^*_{s^\prime} \in \mathtt{S}_s$.
\end{itemize}}
 Afterward, each player evaluates its potential payoff as in \eref{eq5} to make an accurate decision as explained in step II. The selected transmitting player in each coalition is referred to a \textit{coalition head} who can do the analysis in step II. Therefore, this step significantly reduces the search space of the coalition formation.

The coalition formation step optimizes the selection of the transmitting players  and their IDNC packets through many successive split-and-merge rules between coalitions. Therefore, step II is to assign players to potential neighbor coalitions, select the transmitting player, and find its optimal IDNC packet, which can be accomplished by the following. In this step, the time-index
is updated to $\tau=\tau+1$. The merge rules are implemented by checking the merging possibilities of each pair of neighbor coalitions $s$ and $k$. Particularly, a coalition $s\in \Psi_\tau$ can decide to merge with another coalition $k$ to form a new coalition $j$. As such, the resulting structure guarantees both merge conditions (MC).
\begin{itemize}
\item MC1: There exists at least one player satisfies \eref{eq7bb} and \eref{eq7dd}. 
\item MC2: At least one player in the merged coalition can reduce its individual payoff without increasing the payoffs of the other players.
\end{itemize}  
After all the coalitions have made
their merge decisions based on the players preferences, the merge rules end. This results in the updated
coalition structure $\Psi_\tau$. 
Similarly, the split rules performed on the players that do not benefit from being a member of that coalition. In other terms, coalition $s\in \Psi_\tau$ can be splitted into coalitions of smaller sizes as long as the splitter coalitions guarantee both split conditions (SC).
\begin{itemize}
\item SC1: At least one player can strictly enhance its payoff without increasing the payoffs of all the remaining players.
\item SC2: In each split coalition, there exists at least one player satisfying \eref{eq7bb} and \eref{eq7dd}.
\end{itemize} 
At the end of the split rules, the coalition structure $\Psi_1$ is updated. The time index is updated along with a sequence of merge-and-split rules which take place in a distributed manner. Such sequence continues based on the resulting payoff of each player and coalition. It ends when there is no further  merge-and-split rules required in the current coalition structure $\Psi_\tau$, which is the case of the final coalition structure $\Psi_\text{fin}$.

Finally, each transmitting player in each coalition broadcasts an IDNC packet to all players in its coverage zone. The distributed merge-and-split coalition formation algorithm is summarized in Algorithm \ref{Alg1}.
We repeat the above three steps until all packets are disseminated among players, as explained in \algref{Alg2}.

\begin{algorithm}[t!]
\caption{Coalition Formation Distributed Algorithm for a D2D Multi-hop Network}
\label{Alg1}
\textbf{Initialization:}\; 
Players are organized themselves into an initial coalition structure $\Psi_\text{ini}=\{\mathcal{S}_1, \ \cdots, \mathcal{S}_k\}$;\;
Initialize time-index $\tau=0$ and $\Psi_\tau=\Psi_\text{ini}$;\;
\textbf{Step I: Coalition Members Discovery};\;
\begin{itemize}
\item Each player discovers its neighboring players.\;
\For { \text{each} $\mathcal{S}_s\in \Psi_{\text{ini}}, \forall s=\{1,2,\ \cdots, \, k \}$}{
Select the transmitting players $\mathcal{A}_s$ that satisfying \eref{eq7bb} and  \eref{eq7dd} and find $a_s^{*}$ and its IDNC packet $\kappa_{a_s^{*}}$ by solving  \eref{op7}.\;
Calculate the utility of each player as in \eref{eq5}.\;}
\end{itemize}
\textbf{Step II: Coalition Formation};\;
\begin{itemize}
\item The optimization target in coalition $\mathcal{S}_s$ is $\min\limits_{{\substack{\ \underline{a}_t\in\{0,1\}^{|\mathcal{U}_s|} \\\underline{\kappa}\in\{0,1\}^{|\mathcal{M}_s|} }}}\phi(\mathcal{S}_s)$.\;
\item Obtain player's assignments based on the two main rules of merge and split:\;
\Repeat{No further merge nor split rules}{ Update $\tau=\tau+1$.\;
\For { \text{each} $\mathcal{S}_s\in \Psi_{\tau-1}, \forall s=\{1,2,\ \cdots, \, k \}$} {%
The selected transmitting player analyzes all possible merge rules.\;
If a merge occurs, the current coalition structure $\Psi_{\tau-1}$
is updated.\;
Update $\mathcal{A}_s$ and update the selected transmitting player by solving \eref{op7}.\;
Set $\Psi_{\tau}=\Psi_{\tau-1}$.
} 
\For { \text{each} $\mathcal{S}_s\in \Psi_{\tau}, \forall s=\{1,2,\ \cdots, \, k \}$} {%
The selected transmitting player analyzes all possible split rules.\;
If a split occurs, the current coalition structure $\Psi_{\tau}$
is updated.\;
Update $\mathcal{A}_s$ and update the selected transmitting player by solving \eref{op7}.\;
} 

}
\end{itemize}
\textbf{Output} The convergence coalition structure $\Psi_\text{fin}=\Psi_{\tau}$.\;
\textbf{Step III: IDNC Packet Transmission};\;
\begin{itemize}
\item Each transmitting player $a_s^{*}$ in each coalition broadcasts IDNC packet $\kappa_{a_s^{*}}$ to all players in its coverage zone.
\end{itemize}
\end{algorithm}

\begin{algorithm}[t!]
\caption{Overall D2D Multi-hop Approach for Solving Problem \eref{eqct}}
\label{Alg2}
\textbf{Data:} $\mathcal{U}$, $\mathcal{P}$, $\mathcal{H}_u$, $\mathcal{W}_u$, $\mathcal{C}_u$, $\mathcal{T}_u=0$, $\mathcal{D}_u=0$, $\forall~ u\in \mathcal{U}$ and {\boldmath$\epsilon$}.\; 
Set time-index of the completion time $t=0$;\\
\textbf{Repeat:}
\begin{itemize}
\item Execute Algorithm \ref{Alg1} and obtain the IDNC packet for each transmitting player in $\Psi_{\text{fin}}$;
\item Each targeted player does an XOR binary operation and calculate the anticipated completion time as in \eref{eq85}.
\item Each targeted player broadcasts a one bit ACK, indicating the successful reception of the packet, to all players in its coverage zone.
\item $t=t+1$;
\end{itemize}
\textbf{Until} all packets are disseminated among players.\; 
\textbf{Output} the completion time $t$. 
\end{algorithm}

\begin{figure}[t]
\centering
\includegraphics[width=0.45\linewidth]{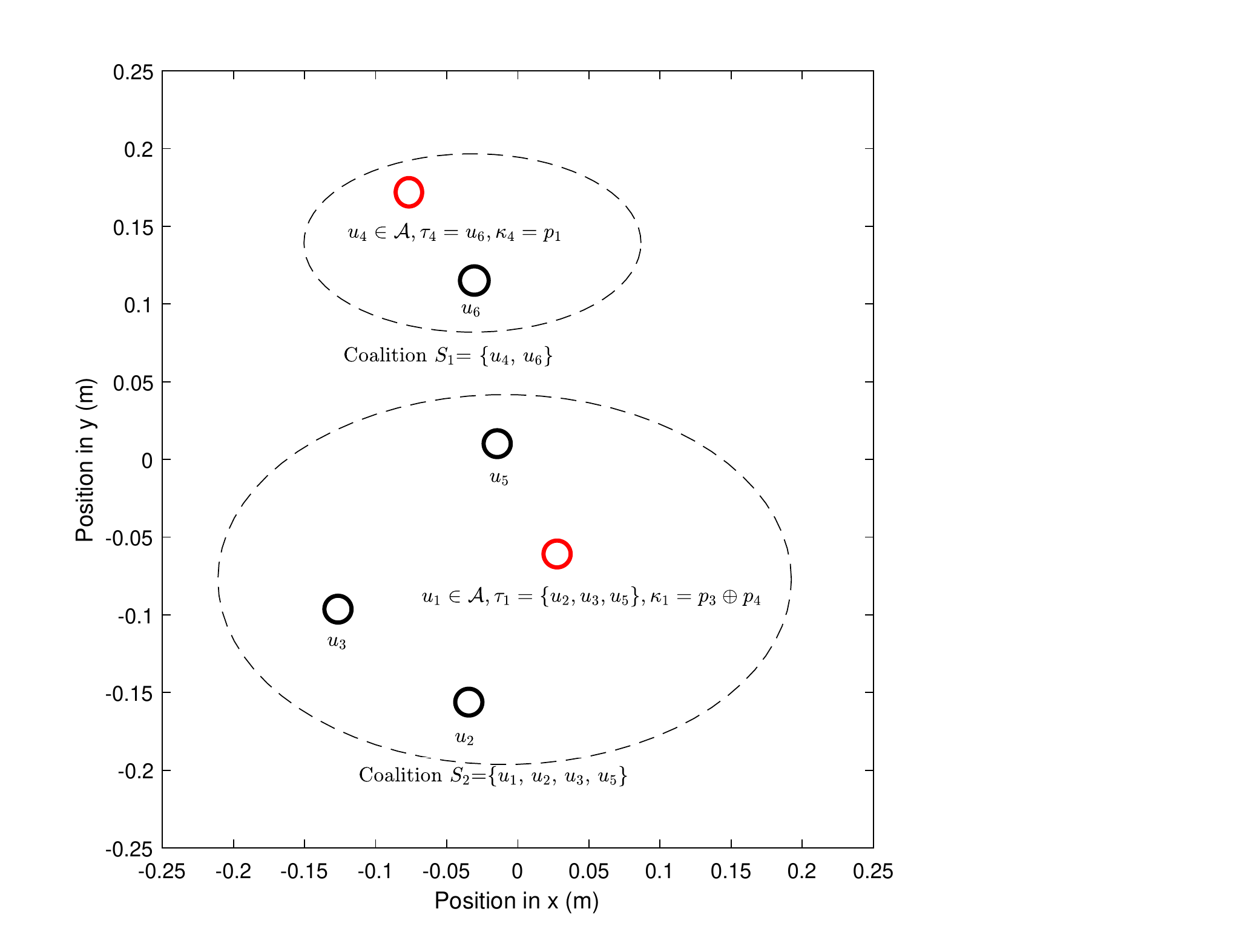}
\caption{A resulting coalition structure $\Psi_\text{fin}=\{\mathcal{S}_1, \mathcal{S}_2\}$ from \algref{Alg1} for a partially connected D2D network that is presented in Fig. \ref{fig1}. }
\label{fig2}
\end{figure}

Fig. \ref{fig2} depicts a snapshot of the coalition structure
$\Psi_\text{fin}$= $\{\mathcal{S}_1, \mathcal{S}_2\}$ resulting from  \algref{Alg1} for the simple D2D network presented in Fig. \ref{fig1}. For ease of analysis, we assume error-free transmissions. Given the coverage zone of each player and their side information as in Fig. \ref{fig1}, two disjoint coalitions are formed where only one player transmits in each coalition. In particular, in coalition $\mathcal{S}_1$, player $u_4$ transmits packet $p_1$ to player $u_6$, and in coalition $\mathcal{S}_2$,  player $u_{1}$ transmits an IDNC packet $p_3\oplus p_4$ to players $u_2$, $u_3$, $u_5$. The transmitting player in each coalition is shown in a red circle; their targeted players and the optimal IDNC packets are shown in Fig. \ref{fig2}. In a nutshell, we shed some remarks on executing \algref{Alg1}.
\begin{itemize}
\item The merge-and-split rules enumerate only the neighbor coalitions, and this does not necessarily need significant computations. To further reduce the computations, the players of a coalition $\mathcal{S}_s$ can avoid merging with other neighbor coalition $\mathcal{S}_k$ if the payoffs of the players in both coalitions are equal $\phi_u(\mathcal{S}_s)=\phi_{u'}(\mathcal{S}_k)$, $\forall u\in  \mathcal{S}_s$ and $\forall u'\in  \mathcal{S}_k$.
\ignore{\item The resulting coalitions' structure from the proposed distributed algorithm $\Psi_{\text{fin}}$ is not always guaranteed to be optimal. This is due to the formed coalitions being unaware of their payoffs and thus have no clue whether any different coalition structure would lead to a lower system payoff. Even if all coalitions' payoffs are known, it is infeasible to find the optimal coalition structure due to complexity constraints \cite{45aa}.}
\item Forming coalitions only one time, i.e., at the first stage of the game, is not guaranteed to disseminate all packets to all players. This is because each formed coalition has only some portion of packets and does not have the wanted packets of other players in other coalitions.  For packet recovery completion, each
coalition is formed, at each transmission round, based on the individual preference of its members and irrespective of the \textit{Has} sets of its members. Thus, each transmitting
player has disseminated some packets to each visited coalition in previous transmissions. 
\end{itemize}

In the considered game, each player has two actions to take either to transmit an IDNC packet $\kappa$ or to listen to a transmission. Therefore, the action of a player $u$ at each game stage $t$ is $\mathcal{AC}_u(t)=\{\text{transmit} ~\kappa_u,~ \text{remain silent}\}$. The asymmetry of the side information at each player generates a different packet combination to be sent by each player at each transmission round. This causes the asymmetry of the action space
of each player. Also, in each transmission, different players are associated with each coalition. All these make the payoff of each coalition unique.

\section{Convergence Analysis, Complexity, and communication overhead} \label{PA}
 This section first studies the convergence of the coalition formation algorithm and its Nash equilibrium stability. Afterward,\ignore{ we use a well-known metric that characterizes the Nash equilibriums to study the robustness of the coalition formation game.  Then,} the complexity properties of \algref{Alg1} is analyzed, which shows that \algref{Alg1} needs a low signaling overhead.
 \subsection{Convergence and Nash Equilibrium} 
In coalition formation games, the stability of the coalition structures 
corresponds to an equilibrium state known as Nash-equilibrium. This subsection proves that the convergence of the coalition formation algorithm is guaranteed and it is a Nash-stable coalition structure. 

The following theorem demonstrates that \algref{Alg1} terminates in a finite number of iterations.
\begin{theorem} \label{th:1w}
Given any initial coalition structure $\Psi_{\text{ini}}$, the coalition formation step of \algref{Alg1}
maps to a sequence of merge-and-split rules which
converges, in a finite number of iterations, to a final coalition structure $\Psi_{\text{fin}}$ composed of a number of disjoint coalitions.
\end{theorem}
\begin{proof}
To proof this theorem, we need to show that for any merge or split rule, there exists a new coalition structure which results from the coalition formation step of \algref{Alg1}. Starting from any initial coalition structure $\Psi_{\text{ini}}$, the coalition formation step of \algref{Alg1} can be mapped to a sequence of merge/split rules. As per definition $8$ and definition $9$, every merge or split rule transforms the current
coalition structure into another coalition structure, hence we obtain the following sequence of coalition structures
\begin{align}
\Psi_{\text{ini}}\rightarrow \Psi_{1} \rightarrow \Psi_{\text{2}} \rightarrow \ \cdots \ \Psi_{\text{fin}}
\label{eq88}
\end{align}

where $\Psi_{i+1} \triangleright \Psi_{i}$, and $\rightarrow$ indicates the occurrence of a merge-and-split rule. Since the Pareto Order introduced in definition $6$ is irreflexive, transitive and monotonic, a coalition structure cannot be revisited. Given the fact that the number of merge and split rules of a finite set is \textit{finite} and the merge/split operations-coalition structure mapping, the number of coalition structure sequences in \eref{eq88} is finite. Therefore, the sequence in \eref{eq88} always terminates and converge to a final coalition structure
 $ \Psi_{\text{fin}}$.
\end{proof}

\ignore{Now, let us analyze the stability of the final coalition structure  $ \Psi_{\text{fin}}$ which results from  \algref{Alg1}.}
\begin{definition} A coalition structure $\Psi=\{\mathcal{S}_1,\ \cdots \ ,\mathcal{S}_k \}$ is Nash-stable if players have no incentive to leave $\Psi$ through merge-and-split operations.
\end{definition}

This definition implies that any coalition structure
$\Psi$ is considered as a Nash-stable coalition structure if and only if no player has an incentive to move from its current coalition and join another coalition or make an individual decision by performing any merge/split rules. Further, the coalitions in the final coalition structure $\Psi_\text{fin}$ have no incentive to do more merge and split operations. A Nash-stable coalition structure is also an individually stable coalition structure. In general, in a coalition formation game, Nash-stability is a subset of individual stability \cite{44}. Specifically, no player leaves its current coalition through a split rule and form an empty coalition, i.e., no singleton coalition is formed if the following property holds.

\begin{property} \label{pr1}
There exists at least one coalition structure $\Psi$ that satisfies both Nash-stability
and individual stability if and only if $\forall\mathcal{S}_s\in \Psi$ such that $|\mathcal{S}_s|>1$.
\end{property}
\begin{IEEEproof}
This property states that forming a singleton coalition cannot happen. Indeed, since each player cannot send an encoded packet to itself, it believes that a better payoff can be obtained by being a member of any coalition. Further, since the payoff of a non-targeted player in any coalition and a single player-coalition is the same, our proposed algorithm, as mentioned in the previous section, avoids making any merge-and-split rules for equal payoff values. Thus, according to  \algref{Alg1}, a Nash-stable and individual stable coalition structure can be obtained.
\end{IEEEproof}

As a consequence of  Property 2, the final coalition structure $\Psi_\text{fin}$ that results from \algref{Alg1} is $\mathbb{D}_\text{hp}$ stable as the coalitions have no incentive to do further merge-and-split operations. $\mathbb{D}_\text{hp}$  stable is also known as merge-and-split proof \cite{44}. Furthermore, $\Psi_\text{fin}$ can be considered as $\mathbb{D}_\text{c}$ stable. This is because players have no incentive to leave $\Psi_\text{fin}$ and form any other coalitions \cite{26}. 

To illustrate the above concepts, consider the resulting coalition structure $\Psi_\text{fin}=\{\mathcal{S}_1, \mathcal{S}_2\}$ that shown in Fig. \ref{fig2}. The coalition structure $\Psi_\text{fin}$ is Nash-stable as no player has an incentive to leave its current coalition. For example, player $u_5$
has a payoff of $\phi_5(\mathcal{S}_2)=-2$ when being part of the coalition $\mathcal{S}_2 = \{u_1, u_2, u_3, u_5\}$. The payoff $\phi_5(\mathcal{S}_2)$ is calculated as follows. Since player $u_5$ receives an IDNC encoded packet from player $u_1$, it does not experience any decoding delay increases. Thus, by \eref{eq85}, its anticipated completion time is $\mathcal{T}_5(\underline{\rm a}_t) = \frac{|\mathcal{W}^{(0)}_5|+ \mathbb{D}_5(\underline{\rm a}_t)-\mathbb{E}[\sigma_5]}{1-\mathbb{E}[\sigma_5]}=2$, and, by \eref{eq5} its payoff is $-2$.  If player $u_5$ switches to act non-cooperatively and joins $\mathcal{S}_1$, player $u_6$ would be the new transmitting player in $\mathcal{S}_1$. In this case, player $u_5$ will be in the coverage zone of both transmitting players $u_1$ in $\mathcal{S}_2$ and $u_6$ in $\mathcal{S}_1$. Consequently, the payoff of player $u_5$ decreases to $\phi_5(\mathcal{S}_1)=-3$, and the payoff of player $u_6$ decreases from $\phi_6(\mathcal{S}_1)=-3$ to
$\phi_6(\mathcal{S}_1)=-4$. Thus, player $u_5$ does not deviate form its current coalition $\mathcal{S}_2$ and join $\mathcal{S}_1$. Similarly, if players $u_2$ and $u_3$ act non-cooperatively by leaving $\mathcal{S}_2$ and forming a singleton coalition for each, i.e., $\mathcal{S}_3$ and $\mathcal{S}_4$, their payoffs decrease from $\phi_2(\{2\})=-2$ and $\phi_2(\{3\})=-1$ to $\phi_2(\mathcal{S}_3)=-3 $ and $\phi_3(\mathcal{S}_4)=-2$, respectively. Clearly, $\Psi_\text{fin}$ is an individual Nash-stable as it does not have any singleton coalition. Further,  it is both $\mathbb{D}_\text{hp}$ and $\mathbb{D}_\text{c}$ stable as no further merge-and-split operations can be performed by the coalitions and no player has incentive to deviate from $\Psi_\text{fin}$, respectively. 

\ignore{\subsection{Robustness of the Proposed Solution}
Nash Equilibrium is not sufficient to guarantee the robustness of the coalition formation game formulation. To characterize the equilibrium of a game, we use a
well-known metric, introduced in \cite{46},  is called the price of anarchy (POA). 

\begin{definition} The price-of-anarchy (PoA) is defined as the worst-case
efficiency of a Nash Equilibrium among all possible strategies. Such PoA of one coalition and of the whole game at any stage can be expressed, respectively, as follows
\begin{align}
&POA_{\mathcal{S}}=\dfrac{\max_{s\in{\mathcal{ST}}}W(s)}{\min_{s\in{\mathcal{NE}}}W(s)}
\text{ and } POA= \dfrac{\max_{\mathcal{S}\in\Psi}POA_{\mathcal{S}}}{\min_{\mathcal{S}\in\Psi}POA_{\mathcal{S}}},
\end{align}
where $\mathcal{ST}$ represents the set of all possible strategies at the game stage, $\mathcal{NE}$ denotes the set of all NEs at the game stage, $W$ : $\mathcal{S}\rightarrow \mathbb{R}$ is a fairness function.

\end{definition}
The notion of PoA in game theory is used to measure how the efficiency of a game degrades. This game's degradation due to the selfish behavior of its players.
In this paper, the payoff values are the same for all the players, i.e., the completion time in \eref{eq5}. Thus, the PoA can be reformulated as follows
\begin{equation} 
\begin{split}
POA= \dfrac{\max_{s\in{\mathcal{NE}}}\phi(s)}{\min_{s\in{\mathcal{NE}}}\phi(s)}.
\end{split}
\end{equation}}

\subsection{Complexity Analysis and Communication Overhead}
This section analyzes the computational complexity and communication burden of  \algref{Alg1}.\\
\textit{Computational Complexity:} Each player at any game stage needs to find the optimal IDNC packet combination, which depends on the packets that it possesses. Further, since a game with incomplete information, i.e., each player knows only the side information of players in its coverage zone, every player can generate the IDNC packet combinations of all other players in its coverage zone. This allows every player to calculate the payoff function \eref{eq5} of all other players in its coverage zone. 

The complexity of generating an optimal IDNC packet using a maximum weight search method is explained as follows. First, the BS generates the vertcies of $O(NM)$, and it connects them by edges that represent network coding conditions of $O(N^2M)$. Then, the BS executes the maximum weight search method that computes the weight of $O(NM)$ vertcies, and selects maximum $N$ users. Hence, the overall complexity of finding the optimal IDNC packet is $O(NM)+O(N^2M)+O(NM*N)=O(N^2M)$ \cite{12m}. In our case, the complexity is bounded by $O(N^2M)$ since the number of players in the coverage zone of each player is less than the total number of players.

\textit{ Communication Overhead:} \ignore{While in the complexity part we analyzed the computation burden for each player to find the IDNC packet, the communication overhead part analyzes the signaling overhead that needs to run Algorithm \ref{Alg1}.} The communication overhead of \algref{Alg1} is related to perform the members' discovery step, coalition heads selection, and the analysis of merge-and-split rules, which is associated with the total number of coalition formations. 

First, similar to many algorithms in the literature, e.g., \cite{44aa}, the member discovery step needs $|N|$ 2-byte messages, in which each message is being sent to all neighbor players which is denoted by $\mathbf{U}$. Thus, the total communication overhead for discovering the neighbor players is $|2N\mathbf{U}|$ bytes.

Second, coalition head selection can be performed in many different strategies, e.g.,  based on players' attributes \cite{50m}, \cite{51m}. In \algref{Alg1}, players in each coalition initially select their coalition head by exchanging an advertisement message among them, and the one that satisfies the conditions C1 and C2 in \sref{DC} would be chosen. The same process is applied for selecting/updating the coalition head in step III. Being a player connected to most players in the coalition, the coalition head is responsible for ensuring that the rest of the coalition's members received an acknowledgment (ACK). As such, they can update their side information after each D2D transmission.

Third, the communication overhead of the coalition formation step is based on the number of merge-and-split rules, which is mainly related to the total number of decisions made by each of the $N$ players. As previously mentioned, the merge-and-split operations enumerate only the neighbor coalitions $\mathtt{S}_s$.\ignore{ Therefore, most players tend to form coalitions of sizes less than $N/2$.} Thus, two extreme cases can occur.
\begin{itemize} 
\item If all coalitions' players decide to leave their current coalitions and join other coalitions. In this case, each player $u$ in coalition $\mathcal{S}_s$ would make $|\mathtt{S}_s|$ decisions (player $u$ has an $|\mathtt{S}_s|$ possibilities to join any of the neighbor coalitions). Consequently, the total number of players' decisions is  $\mathcal{Q}_{\text{worst}}=N|\mathtt{S}_s|$, and the overhead complexity is of the order $O(N|\mathtt{S}_s|)$.

\item If players did not make any decisions.  Since no decision is made by players, the overhead in this case is only $\mathcal{Q}_{\text{best}}=N$ (due to the initial player-coalition associations as in step I), and a complexity order of $O(N)$. 
\end{itemize}
In practical, the number of players' decisions is between the above two cases, i.e., $\mathcal{Q}_{\text{best}}\leq\mathcal{Q}\leq \mathcal{Q}_{\text{worst}}$. Hence, if $\mathcal{L}$ average decisions are made by players, then $\mathcal{Q}=N|\mathcal{L}|$ decisions that perform split-and-merge rules are needed until \algref{Alg1} converges.

Therefore, combining all the signaling overhead components, the total overhead is $N(2\mathbf{U}+|\mathcal{L}|)$. Such signaling overhead will add only a few bytes, which are negligible in size compared to the entire packet's size. Furthermore, to update the \textit{Has} and \textit{Wants} sets of players, only the indices of packets needed for the communication between the players, not their contents. Hence, we ignore signaling overhead factor because it is first constant (independent on the completion time and decoding delay) and that its size is negligible. 
\ignore{\begin{table}
       \caption{Numerical parameters}
       \renewcommand{\arraystretch}{0.99} \label{table:parameters}
       \centering
   \begin{tabular}{|p{6.5cm}|p{5cm}  |}    \hline
   
  Solution &  \\ \hline
   Interference-free D2D  & Perfect \\
   \hline
   Point-to-Multi-Point & Static \\ \hline
   
   Fully connected  D2D  & Random \\
   \hline
      One colaition game   & Uniform   \\ 
      \hline
    Proposed scheme   & Uniform   \\ 
         \hline  
   
    \end{tabular}
\end{table}}

\section{Numerical Results} \label{SR}
In this section, we evaluate the performance of our proposed coalition formation game (denoted by CFG partially-connected D2D) to demonstrate its capability of reducing the completion time compare to the baseline schemes. We first introduce the simulation setup and the comparison schemes. Then, the completion time and game performances are investigated, respectively.
\subsection{Simulation Setup}
We consider an IDNC-enabled partially connected D2D network where players are uniformly re-positioned
for each iteration in a $500$m$\times500$m cell with connectivity index
$\mathbf{C}$, which is defined as the ratio of the average number of neighbors to the total number of players $N$. A simple partially connected D2D network setting is plotted
in Fig. \ref{figsm} for the presented example in Fig. \ref{fig1}. The system setting in this paper follows the setup studied in \cite{29m},\cite{30m}. The initial side information $\mathcal{H}_u$ and $\mathcal{W}_u$, $\forall u\in \mathcal{U}$ of players is independently drawn based on their average erasure probability. The short-range communications are more reliable than the BS-player communications \cite{26m}, \cite{28m}. Hence, unless specified, we assume that the player-to-player erasure probability $\sigma$ is half the BS-to-player erasure $\epsilon$ in all simulations, i.e., $\sigma=0.5\epsilon$. Our simulations were implemented using Matlab on a Windows $10$ laptop $2.5$ GHz Intel Core i7 processor and $8$ GB  $1600$ MHz DDR3 RAM. For the sake of comparison, we implement the following schemes.

\ignore{\begin{figure}[t!]
\centering
\includegraphics[width=0.4\linewidth]{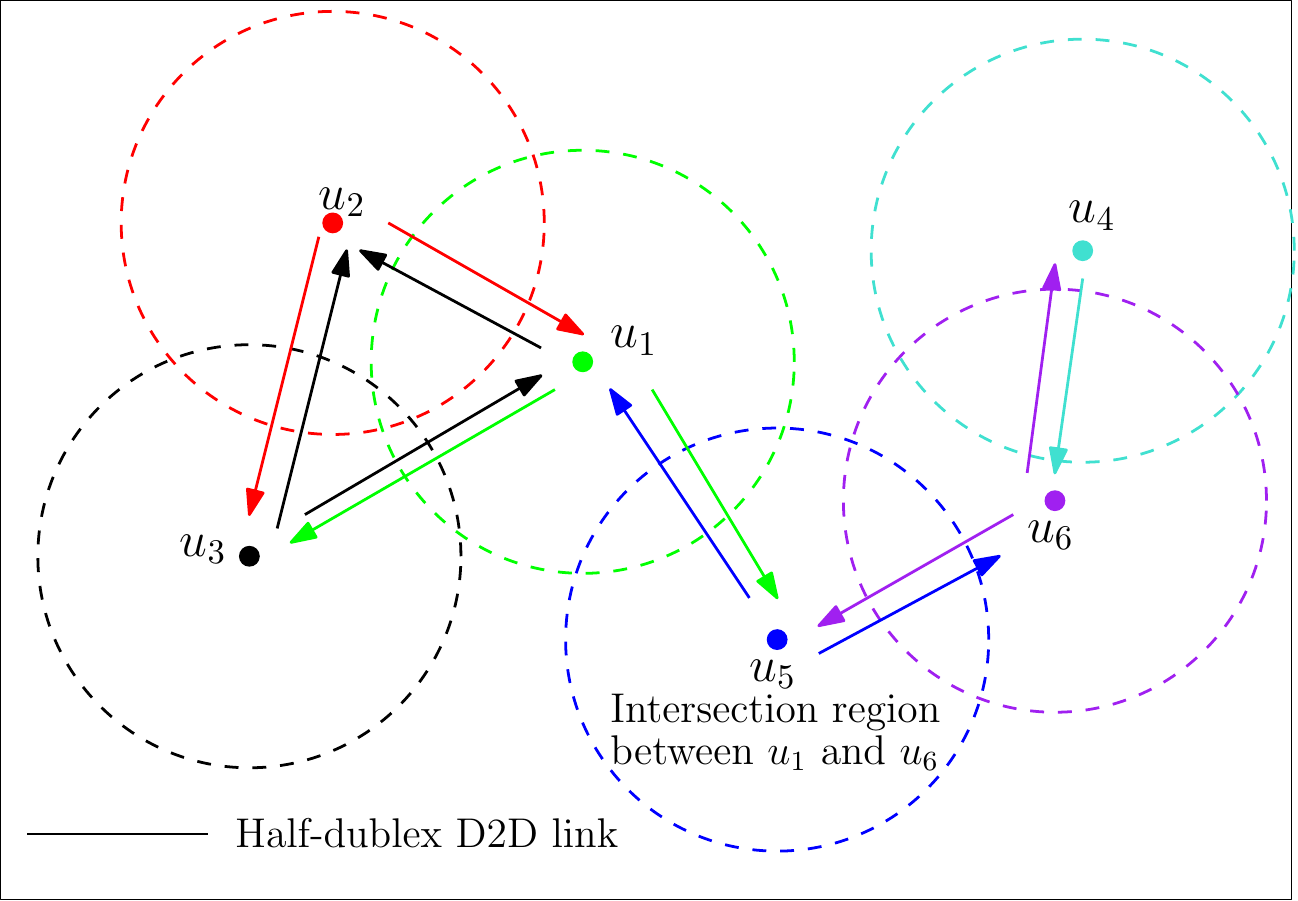}
\caption{A partially connected D2D network of the example presented in Fig. \ref{fig1}.}
\label{figsm}
\end{figure}}

\begin{itemize}
\item The fully-connected D2D system in which a single user who has the largest number of received packets transmits an IDNC packet at each round.
\item The PMP system in which the BS is responsible for the transmissions. The BS holds all the requested packets and can serve all the users. This scheme was proposed in \cite{16m}.
\item The one coalition formation game in a partially connected D2D  (denoted by OCF partially-connected D2D). In this scheme, only one coalition is formed, and a single player transmits an IDNC packet at each round. The transmitting player is selected based on its number of received packets as well as on the maximum number of players in its coverage zone.    
\item The partially D2D in FRANs (denoted by FRAN partially-connected D2D). In this scheme, a fog central unit is responsible for determining the set of transmitting users and the packet combinations. This scheme was proposed in \cite{30m}.
\end{itemize}

\begin{figure}[t!]
    \centering
    \begin{minipage}{0.494\textwidth}
        \centering
   \includegraphics[width=0.7\textwidth]{NT.pdf} 
        \caption{A partially connected D2D network of the example presented in Fig. \ref{fig1}.}
        \label{figsm}
    \end{minipage}\hfill
    \begin{minipage}{0.494\textwidth}
        \centering
    \includegraphics[width=0.7\textwidth]{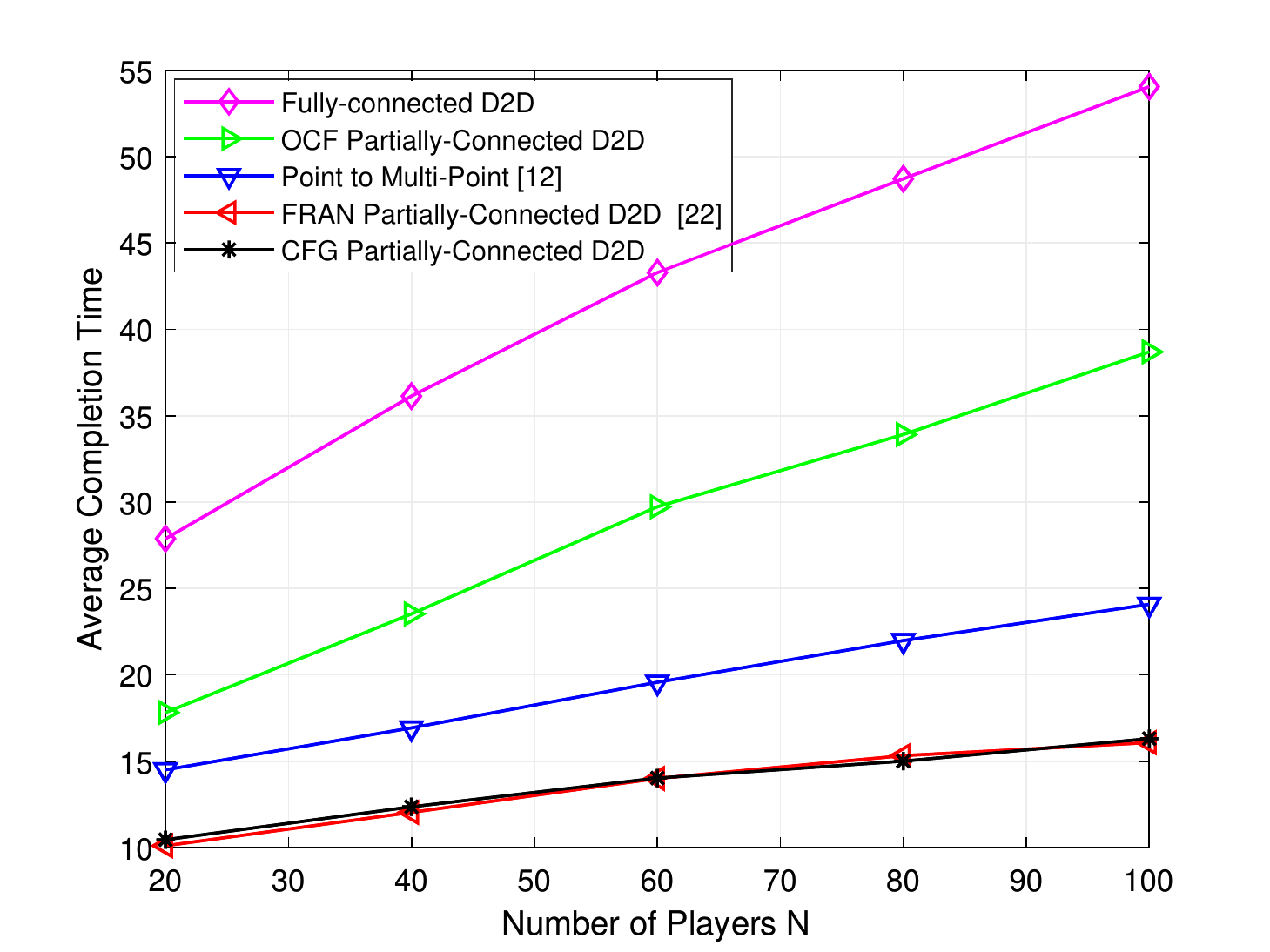} 
        \caption{Average completion time as a function of the number
        of players $N$.}
        \label{fig4}
     \end{minipage}
     \centering
        \centering
            \begin{minipage}{0.494\textwidth}
                \centering
           \includegraphics[width=0.7\textwidth]{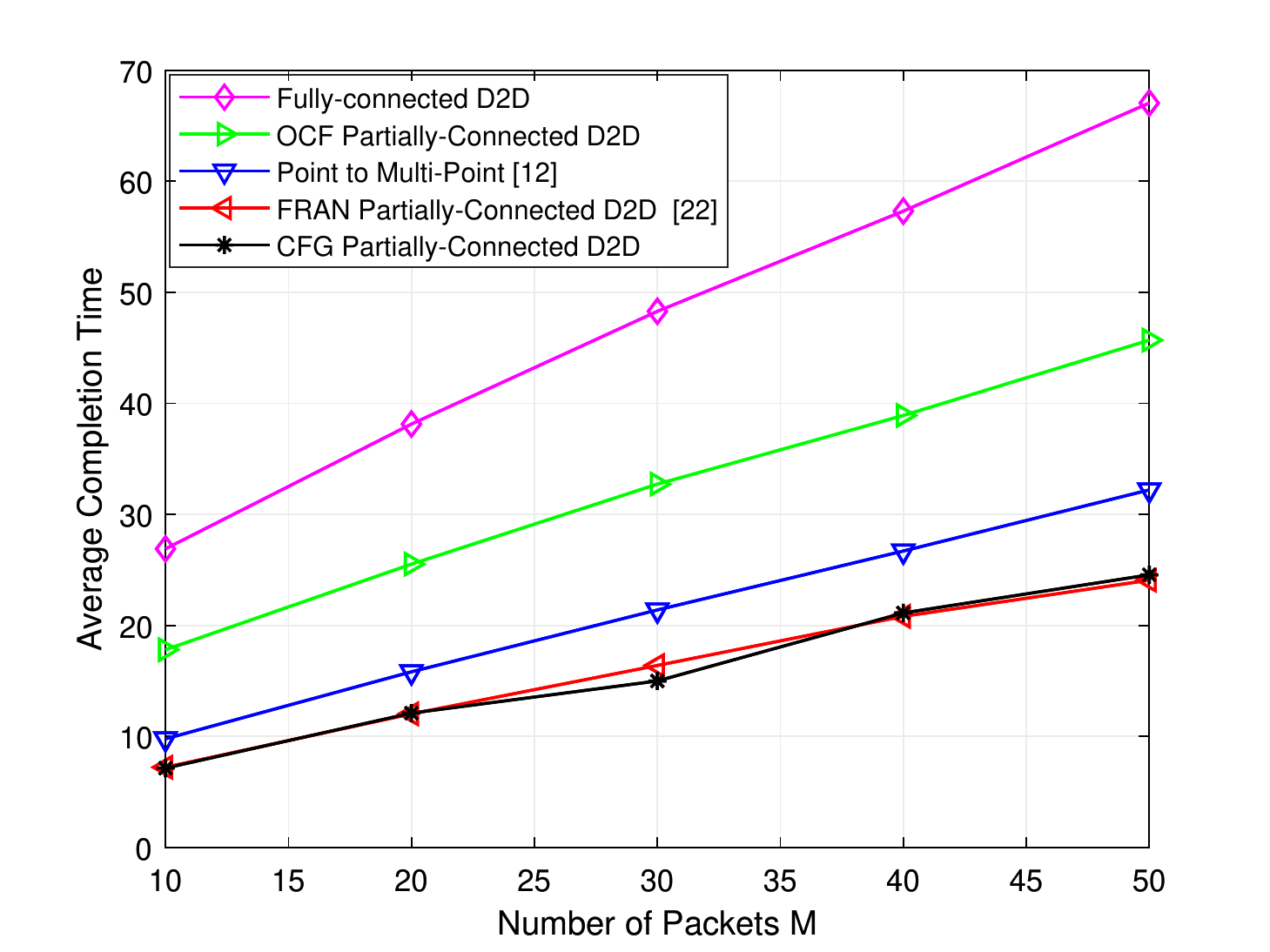} 
                \caption{Average completion time  as a function of the  number of packets $M$.}
                \label{fig5}
            \end{minipage}\hfill
            \begin{minipage}{0.494\textwidth}
                \centering
            \includegraphics[width=0.7\textwidth]{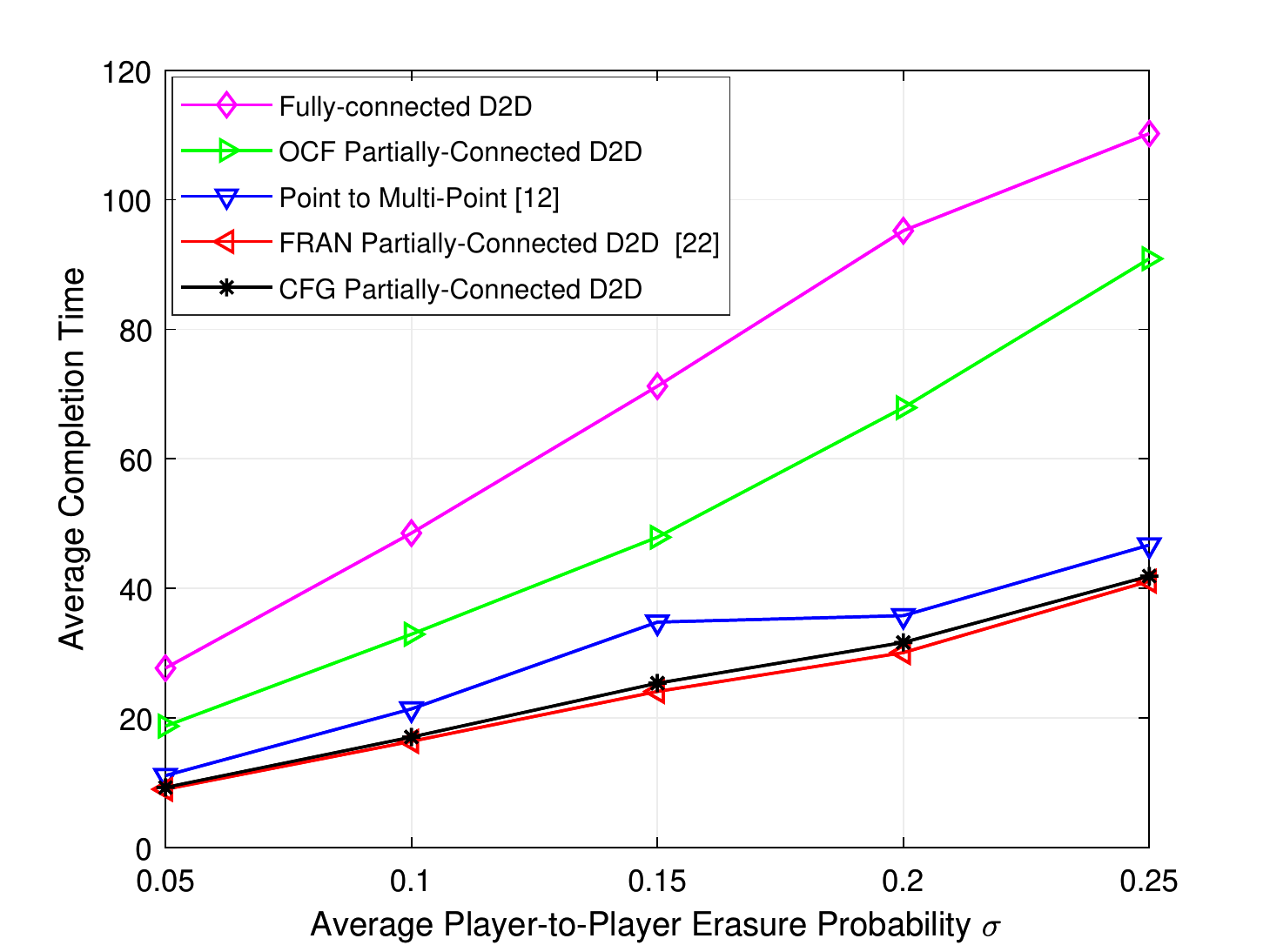} 
                \caption{Average completion time  as a function of the average player-player erasure probability $\sigma$.}
                \label{fig6}
             \end{minipage}
     
\end{figure}
\subsection{Completion Time Performance Evaluation}
To study the completion time performance of the proposed solution, we change the number of players, packets, connectivity index, and the packet erasure probability. \ignore{The completion time is calculated over a
certain number of iterations, and the average values are presented.}

In Fig. \ref{fig4}, we depict the average completion time as a function of the number of players $N$ for a network composed of $M = 30$ packets, $\epsilon=0.25$, $\sigma=0.12$, and connectivity index  $C=0.4$. It is observed from Fig. \ref{fig4} that the proposed CFG partially-connected D2D algorithm outperforms the PMP, fully-connected D2D, and OCF partially-connected D2D schemes for all simulated number of players. This is because of the simultaneous IDNC packet transmissions from cooperating players at the same time. In particular, the fully-connected D2D system only considers the size of the \textit{Has} set as a metric to select a single player for transmission at each round, i.e., $a^*=\max\limits_{a\in \mathcal{U}} \mathcal{H}_a$. The OCF partially-connected D2D scheme focuses on the maximum number of connected players to be formed as well as on the size of the \textit{Has} set of the transmitting player. On the other hand, although the transmitter in the PMP scheme can encode all the IDNC combinations and target a certain number of players, the PMP scheme sacrifices the utility of the simultaneous transmissions by considering only one transmission. Our proposed algorithm strikes a balance between these
aspects by jointly considering the number of targeted players and the \textit{Has} set size of each transmitting player.  Despite the gain achieved by the FRAN partially-connected D2D solution with the presence of a fog that executes the whole process, our decentralized solution reaches the same performance. Clearly, due to the philosophy of the D2D simultaneous transmissions that both schemes have proposed, their performances are roughly the same.

We observe from Fig. \ref{fig4} that, for a small number of players, the PMP system is close to both the CFG partially-connected D2D and FRAN partially-connected D2D schemes. This is because, for a small number of players ($N\leq 60$), the certainty that the whole frame $M$ is distributed between players in the initial transmissions is low, thus decreasing the probability of exchanging potential IDNC packets between players. This makes the overall completion time performance of the partial D2D scenarios close to the PMP scheme. As the number of players increases ($N\geq 80$), the bigger the certainty that the union of their \textit{Has} sets is equal to $M$. This results in more potential D2D IDNC packet exchange, thus increasing the gap between the PMP performance and both the FRAN partially-connected D2D and proposed schemes. 

In Fig. \ref{fig5}, we  illustrate the average completion time as a function of the number of packets $M$ for a network composed of $N= 30$ players, $\epsilon=0.25$, $\sigma=0.12$, and connectivity index $C=0.4$. The figure shows that the proposed scheme
outperforms the fully connected, one coalition game, and PMP schemes. For a few packets, the IDNC combinations
are limited which affect the ability of the
proposed scheme to generate coded packets that satisfy number of players. With increasing the number of packets, the number of transmissions needed for the completion for the aforementioned schemes is remarkably increasing. Therefore, as the number of packets
increases, the proposed scheme outperforms largely the fully connected and one coalition game schemes. We see from Fig. \ref{fig4} that the completion time of all schemes linearly increases with the number of packets. This is expected as the number of packets increases, a high number of transmissions is required towards the completion. This results in increasing the average completion time.

\ignore{\begin{figure}
    \centering
    \begin{minipage}{0.494\textwidth}
        \centering
      \includegraphics[width=0.71\textwidth]{./fig/figp01.eps} 
        \caption{Average completion time  as a function of the average player-player erasure probability $\sigma$ for a poorly connected network.}
        \label{fig7}
    \end{minipage}\hfill
    \begin{minipage}{0.494\textwidth}
        \centering
       \includegraphics[width=0.71\textwidth]{./fig/figp04.eps} 
        \caption{Average completion time  as a function of the average player-player erasure probability $\sigma$ for a moderately connected network.}
        \label{fig8}
     \end{minipage}
\end{figure}}

In Fig. \ref{fig6}, we plot the average completion time as a function of the average player-player erasure probability $\sigma$ for a network composed of $N= 60$, $M=30$, $\epsilon=2\sigma$, and $C=0.4$. Similar to what we have discussed in the above figures, the average completion time of the partial D2D solutions is noticeable compared to the fully-connected D2D and OCF partially-connected D2D schemes, as shown in Fig. \ref{fig6}. We clearly see that the completion time of the partial D2D schemes is better than the PMP one because of their multiple players' transmissions at each round. Moreover, as the player-to-player erasure probability increases, the BS-player erasure probability increases two-fold ($\epsilon=2\sigma$), thus slightly affecting the performance of the PMP scheme. The partial D2D settings, however, benefit from short range and reliable communications which provide much better players reachability and IDNC packet successful delivery compared to the PMP setting. 

\begin{figure}[t!]
    \centering
    \begin{minipage}{0.494\textwidth}
        \centering
       \includegraphics[width=0.61\textwidth]{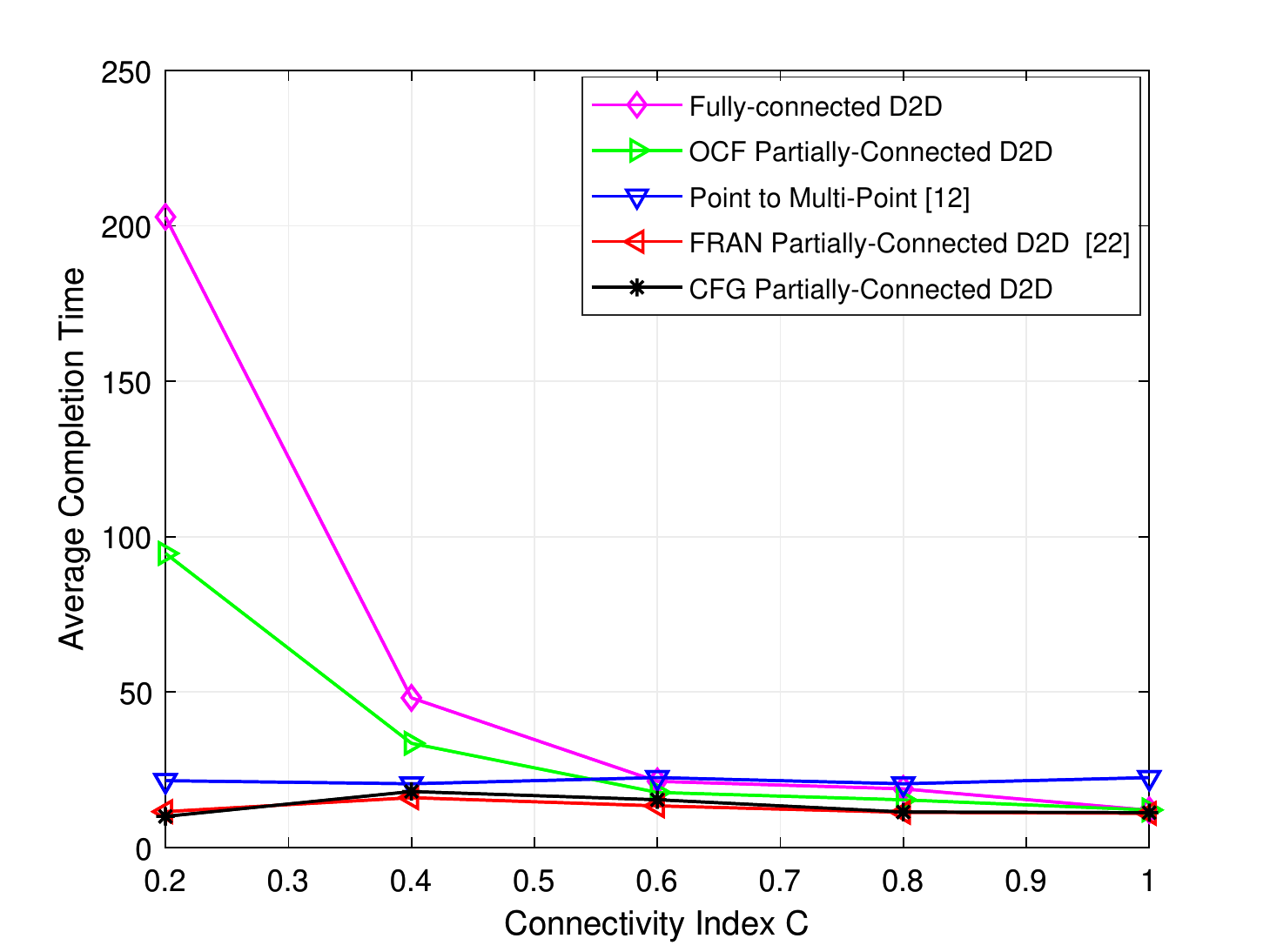} 
        \caption{Average completion time  as a function of the connectivity index $C$.}
        \label{fig9}
    \end{minipage}\hfill
    \begin{minipage}{0.494\textwidth}
        \centering
       \includegraphics[width=0.61\textwidth]{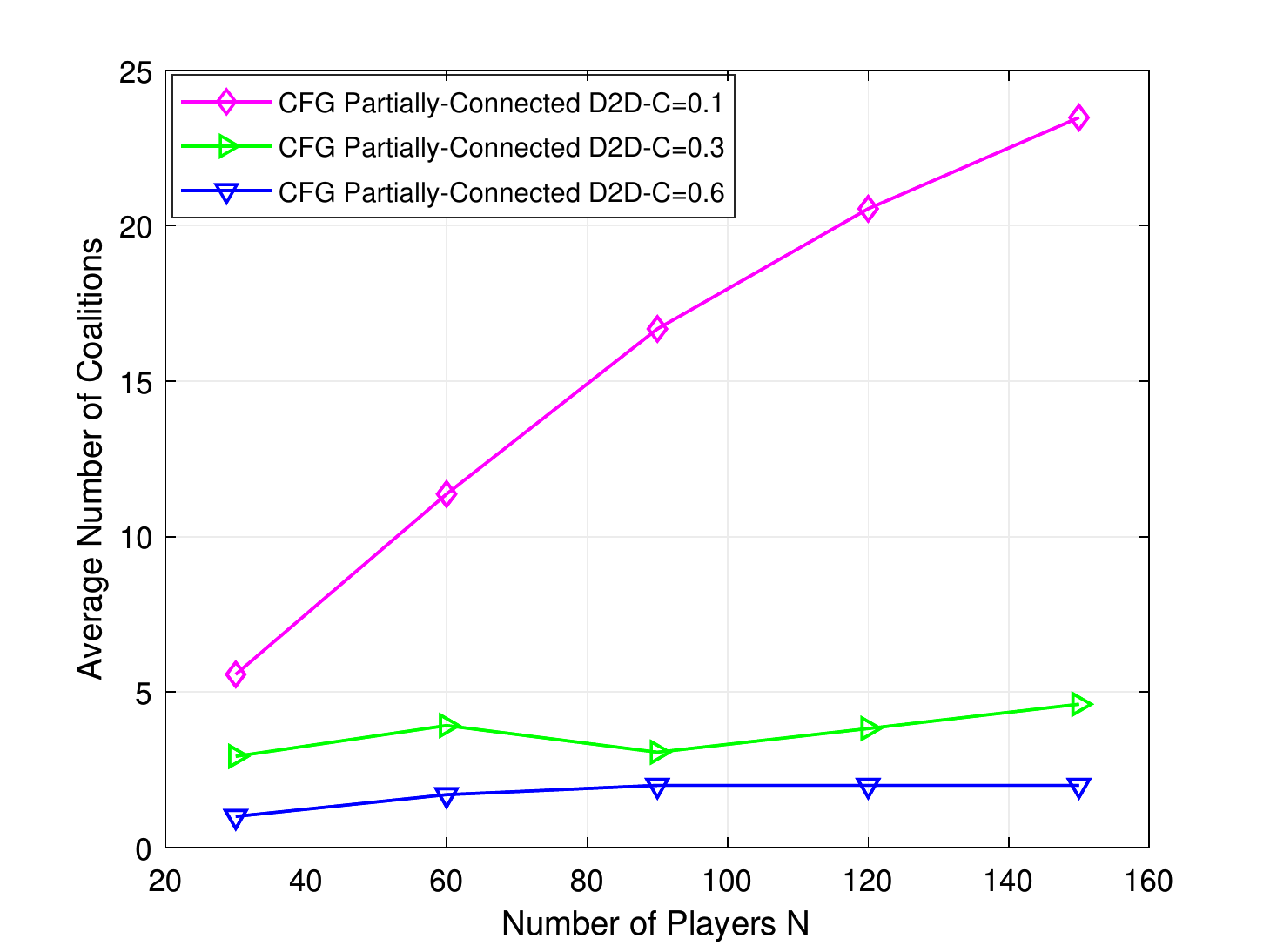} 
        \caption{Average number of coalitions  as a function of the number of players $N$.}
        \label{fig10}
     \end{minipage}
\end{figure}

In Fig. \ref{fig9}, we investigate the average completion time as a function of the connectivity index $C$
for a network composed of $N= 60$, $M=30$, $\epsilon=0.25$, and $\sigma=0.12$. It can clearly be seen that for a low connectivity index $(C \leq 0.4)$, the proposed CFG partially-connected D2D approach noticeably outperforms the fully-connected D2D and OCF partially-connected D2D approaches. In such poorly connected networks $(C \leq 0.4)$, multiple simultaneous players' transmissions are exploited in partially D2D algorithms. However, as the connectivity index increases $(C \geq 0.6)$, the number of formed disjoint coalitions in our proposed solution is drastically reduced, thus reducing the number of transmitting players. This results in a performance agreement with the fully-connected D2D scheme.  Being independent of the coverage zones of the transmitting players and the delay created by those players, the PMP scheme is not affected by the changes to $C$. Thus, the PMP scheme has constant average completion time.
\ignore{

For moderately connected networks in Fig. \ref{fig4} and Fig. \ref{fig6}, the proposed CFG partially-connected D2D still outperforms the fully-connected D2D and the OCF partially-connected D2D schemes for all simulated numbers of players and packets. Due to the high connectivity between players, the number of formed coalitions is small, thus decreasing the number of transmitting players in partial D2D scenarios.  Consequently, the completion time performance of the partially connected networks, i.e., CFG partially-connected D2D and FRAN partially-connected D2D, is degraded and expected to be close to the performance of the fully-connected D2D and OCF partially-connected D2D schemes.}

\begin{table}
       \caption{The influence of changing $\sigma$ on the completion time performance of the proposed scheme}
       \renewcommand{\arraystretch}{0.8} \label{table1}
       \centering
\tabcolsep=0.11cm
     \begin{tabular}{|p{5.2cm}|p{1.6cm}  |p{1.6cm} |p{1.6cm}| p{1.6cm}|}    \hline
    
    Solution & $\sigma=0.6\epsilon$ & $\sigma=0.7\epsilon$ & $\sigma=0.9\epsilon$ & $\sigma=\epsilon$  \\ \hline
   Point to Multi-Point & 30.2900 & 30.2800 & 30.3100 & 30.4800
      \\
     \hline
     CFG partially-connected D2D & 20.1800& 23.4702 & 30.4500& 33.9300\\ \hline
    
     \end{tabular}
  
\end{table}

To conclude this section, we study the influence of the setting $\sigma=0.5\epsilon$ on the completion time performance of our proposed scheme. In \tref{table1}, we summarize the completion time perfromance for different values of $\sigma$. The considered network setup has $30$ players,
$20$ packets, $\epsilon=0.5$, and $C=0.1$. From \tref{table1}, we note that the completion time of our proposed solution still outperforms the PMP scheme for $\sigma=0.7\epsilon$ and approximately reaches the same performance as for the PMP scheme for $\sigma=0.9\epsilon$. This is due to the simultaneous transmissions and cooperative decisions by the transmitting players, which show the potential of the proposed CFG solution in minimizing the completion time of users.
\subsection{Proposed CFG Perfromance Evaluation}

To quantify the analysis of the proposed formation coalition solution, we plot in Fig. \ref{fig10} the average number of coalitions as a function of the number of players $N$ for a network composed of $M= 30$, a different connectivity index  ($C= 0.6$, $C= 0.3$, and $C= 0.1$), and $\sigma=0.12$. Fig. \ref{fig10} shows that the average coalition size increases with the increase in the number of players. This is because, as $N$ increases, the number of cooperating players increases, thus increasing the average size of the formed coalitions. We can conclude from Fig. $\ref{fig10}$ that the resulting coalition structure $\Psi_\text{fin}$ from Algorithm \ref{Alg1} is composed of a small number of relatively large coalitions when $C= 0.6$. When $C= 0.1$, this number of formed coalitions increase and the resulting coalition structure $\Psi_\text{fin}$ is composed of a large number of small coalitions' sizes.

\begin{table}
       \caption{Average Running Times of the different schemes}
       \renewcommand{\arraystretch}{0.8} \label{table2}
       \centering
       \tabcolsep=0.11cm
  \begin{tabular}{|p{5.4cm}|p{4.3cm}  |p{4.3cm}  |}    \hline
    
   Solution &  Time(s)- Small network& Time(s)- Large network \\ \hline
      FRAN partially-connected D2D  & 0.561893 & 15.98450\\
      \hline
      Point to Multi-Point & 1.994500 & 1103.020716 \\ \hline
      
      Fully-connected D2D  & 0.756420 & 128.772580 \\
      \hline
         OCF partially-connected D2D   & 0.783575 & 28.726515   \\ 
         \hline
       CFG partially-connected D2D   & 0.736737 & 21.725739   \\ 
            \hline

     \end{tabular}
\end{table}
In \tref{table2}, we evaluate the complexity of the proposed coalition game solution as a function of the algorithmic running time. In particular, \tref{table2} lists the consumed time of MATLAB to execute all schemes in different network setups since starting the algorithms until all players receive their wanted packets. The considered small network setup has $30$ players,
$20$ packets, $\epsilon=0.5$, $\sigma=0.25$, and $C=0.1$. The considered large network setup has $100$ players,
$70$ packets, $\epsilon=0.5$, $\sigma=0.25$, and $C=0.1$. It can clearly be seen from the table that the proposed CFG-partially D2D scheme needs low consumed time than all other solutions for both network setups. Although the completion time achieved by the CFG partially-connected D2D scheme is roughly the same as the centralized FRAN partially-connected D2D, the computing time required by our developed scheme is slightly higher than that required by the FRAN partially-connected D2D. This is because our proposed scheme needs time to converge before generating the output. The centralized FRAN scheme has low execution time due to the presence of the fog entity.

Finally, to evaluate the convergence rate analysis of the proposed scheme, the average number of merge-and-split
rules before Algorithm~\ref{Alg1} converges to the final coalition structure is listed in \tref{table3}. To achieve the stable coalition with our proposed CFG scheme, network setup $1$ requires on average $16$ iterations, and network setup $2$ needs on average $22$ iterations.  These results show that our proposed distributed algorithm is robust to different network setups. In summary, these results show that our proposed
algorithm allows D2D users to form stable coalitions with a good convergence speed, which further confirm the theoretical findings in \thref{th:1w}.
\begin{table}
       \caption{Average Number of Coalitions and Split/merge rules of the proposed scheme in the first iteration}
       \renewcommand{\arraystretch}{0.8} \label{table3}
       \centering
   \begin{tabular}{|p{5.5cm} |p{3.7cm}  |p{3.7cm}  |}    \hline
  
  Network Setup & Number of Coalitions & Split-and-merge rules \\ \hline
   Setup 1: $N=100$ and $C=0.1$  & 16.34 & 8.12 \\
   \hline
   Setup 2: $N=160$ and $C=0.1$ &  23.67 & 12.76 \\ \hline
   \end{tabular}
\end{table}

\section{Conclusion} \label{CC}
This paper has developed a distributed game-theoretical framework for a partially connected D2D network using coalition game and IDNC optimization. As such, the completion time of users is minimized. In particular, our proposed model is formulated as a coalition formation game with nontransferable utility, and a fully distributed coalition formation algorithm is proposed. The proposed distributed algorithm is converged to a Nash-stable coalition structure using split-and-merge rules while accounting for the altruistic players' preferences. With such a distributed solution, each player has to maintain a partial feedback matrix only for the players in its coverage zone instead of the global feedback matrix required in the
fully connected D2D networks. A comprehensive completion time and game performances evaluation have been carried out for the proposed distributed coalition game. In particular, our performance evaluation results comprehensively demonstrated that our proposed  distributed solution offers almost same completion time performance similar to centralized FRAN D2D network. \ignore{ Numerical results demonstrate that the proposed distributed solution provides appreciable completion time performance gains compared to the conventional point-to-multipoint and fully connected D2D networks. Compared to the centralized FRAN D2D network, our proposed solution offers roughly the same performance. }

\end{document}